\documentclass[10pt]{article}

\usepackage{latexsym}
\usepackage{theorem}
\usepackage{color,graphicx}
\usepackage{amssymb}
\usepackage{amsmath}
\usepackage{fullpage}

\newenvironment{proof}[1][Proof]
{\par\noindent{\bf #1:} }{\hspace*{\fill}\nolinebreak{$\Box$}\bigskip\par}
\newcommand{\qed}{\hspace*{\fill}\nolinebreak\ensuremath{\Box}}

\newcommand{\f}[3][0pt]{%
\begin{list}{}{%
 \setlength{\leftmargin}{#2 em}%
 \addtolength{\leftmargin}{#2 em}%
 \if!#1!\else\addtolength{\leftmargin}{#1}\fi%
 \setlength{\topsep}{2pt}
 \setlength{\partopsep}{0pt}%
 \if!#1!\else\setlength{\itemindent}{-#1}\fi%
}%
     \item[] #3%
\end{list}%
}

\newtheorem{theorem}{Theorem}
\newtheorem{lemma}{Lemma}

\newtheorem{observation}{Observation}
\theorembodyfont{\upshape}
\newtheorem{defn}{Definition}

\newcommand{\AlgSimpleConnectedPathwidth}{{\tt SCP}}
\newcommand{\AlgConnectedPathwidthWithHomebase}{{\tt CPH}}
\newcommand{\AlgProcessLeftBranch}{{\tt PLB}}
\newcommand{\AlgProcessRightBranch}{{\tt PRB}}
\newcommand{\AlgConnectedPathwidth}{{\tt CP}}

\newcommand{\border}{\delta}
\newcommand{\lborder}{\border_{\textup{L}}}
\newcommand{\rborder}{\border_{\textup{R}}}
\newcommand{\branch}{\mathcal{B}}
\newcommand{\lbranch}{\branch_{\textup{L}}}
\newcommand{\rbranch}{\branch_{\textup{R}}}
\newcommand{\Ext}{\textup{Ext}}

\newcommand{\EL}{{\tt LE}}
\newcommand{\ER}{{\tt RE}}
\newcommand{\algrule}[1][\linewidth]{\vspace*{2pt}\hrule width #1\vspace*{4pt}\noindent}

\newcommand{\pw}{\textup{pw}}
\newcommand{\cpw}{\textup{cpw}}
\newcommand{\tw}{\textup{tw}}
\newcommand{\ctw}{\textup{ctw}}
\newcommand{\width}{\textup{width}}

\newcommand{\cP}{\mathcal{P}}
\newcommand{\cG}{\mathcal{G}}
\newcommand{\cC}{\mathcal{C}}

\newcommand{\wC}{\widetilde{C}}

\newcommand{\sn}{\textup{\texttt{s}}}
\newcommand{\msn}{\textup{\texttt{ms}}}
\newcommand{\nsn}{\textup{\texttt{ns}}}
\newcommand{\mnsn}{\textup{\texttt{mns}}}

\newcommand{\csn}{\textup{\texttt{cs}}}
\newcommand{\mcsn}{\textup{\texttt{mcs}}}

\begin{document}

\title{From Pathwidth to Connected Pathwidth\thanks{An extended abstract of this work appeared in the proceedings of the
28th Symposium on Theoretical Aspects of Computer Science (STACS) 2011.}}

\author{Dariusz Dereniowski\thanks{\noindent Partially supported by MNiSW grant N~N206~379337 (2009-2011).}\vspace*{1cm}\\
Department of Algorithms and System Modeling,\\
Gda\'{n}sk University of Technology, Poland,\\
email: deren@eti.pg.gda.pl}

\date{}
\maketitle

\begin{center}
\parbox[c]{12 cm}{
\begin{center}
\textbf{Abstract}
\end{center}

It is proven that the connected pathwidth of any graph $G$ is at most $2\cdot\pw(G)+1$, where $\pw(G)$ is the pathwidth of $G$. The method is constructive, i.e. it yields an efficient algorithm that for a given path decomposition of width $k$ computes a connected path decomposition of width at most $2k+1$. The running time of the algorithm is $O(dk^2)$, where $d$ is the number of `bags' in the input path decomposition.

The motivation for studying connected path decompositions comes from the connection between the pathwidth and the search number of a graph. One of the advantages of the above bound for connected pathwidth is an inequality $\csn(G)\leq 2\sn(G)+3$, where $\csn(G)$ and $\sn(G)$ are the connected search number and the search number of $G$. Moreover, the algorithm presented in this work can be used to convert a given search strategy using $k$ searchers into a (monotone) connected one using $2k+3$ searchers and starting at an arbitrary homebase.
}

\vspace*{20 pt}

\end{center}

\newpage

\section{Introduction}
\label{sec:intro}

The notions of pathwidth and treewidth are receiving increasing interest since the series of Graph Minor articles by Robertson and Seymour, starting with~\cite{RobertsonSeymour83}. The importance of those parameters is due to their numerous practical applications, connections with several graph parameters and usefulness in designing graph algorithms. Informally speaking, the \emph{pathwidth} of a graph $G$, denoted by $\pw(G)$, says how closely $G$ is related to a path. Moreover, a path decomposition captures the linear path-like structure of $G$. (For a definition see Section~\ref{sec:definitions}.)

Here we briefly describe a graph searching game that is one of the main motivations for the results presented in this paper. A team of $k$ searchers is given and the goal is to capture an invisible and fast fugitive located in a given graph $G$. The fugitive also has the complete knowledge about the graph and about the strategy used by the searchers, and therefore he will avoid being captured as long as possible. The fugitive is captured when a searcher reaches his location. In this setting the game is equivalent to the problem of clearing all edges of a graph that is initially entirely contaminated. There are two main types of this graph searching problem. In the \emph{node searching} two moves are allowed: placing a searcher on a vertex and removing a searcher from a vertex. An edge becomes clear whenever both of its endpoints are simultaneously occupied by searchers. In the \emph{edge searching} we have, besides to the two mentioned moves, a move of sliding a searcher along an edge. In this model an edge $\{u,v\}$ becomes clear if a searcher slides from $u$ to $v$ and either all other edges incident to $u$ have been previously cleared, or another searcher occupies $u$. In both models the goal is to find a search strategy (a sequence of moves of the searchers) that clears all the edges of $G$. The \emph{node} (\emph{edge}) \emph{search number} of $G$, denoted by $\nsn(G)$ ($\sn(G)$, respectively), equals the minimum number of searchers sufficient to construct a node (edge, respectively) search strategy. An important property is that $\pw(G)=\nsn(G)-1$ for any graph $G$~\cite{Kinnersley92,KirousisPapadimitriou85,searching_and_pebbling,Mohring90}. The edge searching problem is closely related to node searching, i.e. $|\sn(G)-\nsn(G)|\leq 1$~\cite{Bienstock_survey91}, and consequently to pathwidth, $\pw(G)\leq\sn(G)\leq\pw(G)+2$.

In this work we are interested in special types of path decompositions called connected path decompositions. The motivation comes from the need of constructing connected search strategies. An edge search strategy is \emph{connected} if the subgraph of $G$ that is clear is always connected. The minimum number of searchers sufficient to construct a connected (edge) search strategy, denoted by $\csn(G)$, is the \emph{connected search number} of $G$. This model of graph searching receives recently growing interest, because in many applications the connectedness is a requirement.

The concept of recontamination plays an important role in the field of graph searching problems. If the fugitive is able to reach an edge that has been previously cleared, then we say that the edge becomes \emph{recontaminated}. If no recontamination occurs during a search strategy, then the strategy is \emph{monotone}. The minimum number of searchers needed to construct a monotone edge (node, or connected, respectively) search strategy is denoted by $\msn(G)$ ($\mnsn(G)$, $\mcsn(G)$, respectively). For most graph searching models it is proven that there exists a monotone search strategy using the minimum number of searchers, in particular $\sn(G)=\msn(G)$~\cite{monotonicity_in_graph_searching,LaPaugh93}, which carries over to node searching, $\nsn(G)=\mnsn(G)$~\cite{searching_and_pebbling} for any graph $G$. In the case of connected graph searching problem it turns out that `recontamination does help' to search a graph~\cite{sweeping_large_cliques}, that is, there exist graphs $G$ for which each monotone search strategy requires more searchers than some non-monotone search strategies~\cite{sweeping_large_cliques}, i.e. $\mcsn(G)>\csn(G)$. For surveys on graph decompositions and graph searching problems see e.g. \cite{searchng_and_sweeping,Bienstock_survey91,treewidth_guide,guaranteed_graph_searching}.

\subsection{Related work}
There are several results that give a relation between the connected and the `classical' search numbers of a graph. Fomin et al. proved in~\cite{price_of_connectedness} that the connected search number of an $n$-node graph of branchwidth $b$ is bounded by $O(b\log n)$ and this bound is tight. One of the implications of this result is that $\csn(G)=O(\log n)\pw(G)$. Nisse proved in~\cite{connected_searching_chordal_graphs} that $\csn(G)\leq(\tw(G)+2)(2\sn(G)-1)$ for any chordal graph $G$. Barri\`{e}re et al. obtained in~\cite{Barriere_Inria2010} a constant upper bound for trees, namely for each tree $T$, $\csn(T)/\sn(T)<2$.
On the other hand, there exists an infinite family of graphs $G_k$ such that $\csn(G_k)/\sn(G_k)$ approaches $2$ when $k$ goes to infinity~\cite{searching_not_jumping}. In this work we improve the previously-known bounds for general graphs by proving that $\csn(G)\leq\sn(G)+3$ for any graph $G$.

Fraigniaud and Nisse presented in \cite{connected_treewidth} a $O(nk^3)$-time algorithm that takes a width $k$ tree decomposition of a graph and returns a connected tree decomposition of the same width. (For definition of treewidth see e.g. \cite{treewidth_guide,RS_treewidth}.) Therefore, $\tw(G)=\ctw(G)$ for any graph $G$. That result also yields an upper bound of $\csn(G)\leq(\log n+1)\sn(G)$ for any graph $G$.

The problems of computing the pathwidth (the search number) and the connected pathwidth (the connected search number) are NP-hard, also for several special classes of graphs, see e.g.~\cite{deren_weighted_trees,Gustedt93,pathwidth_hard,MegiddoHGJP88,weighted_pathwidth09,Peng06}.

\subsection{This work}
This paper presents an efficient algorithm that takes a (connected) graph $G$ and its path decomposition $\cP=(X_1,\ldots,X_d)$ of width $k$ as an input and returns a connected path decomposition $\cC=(Z_1,\ldots,Z_m)$ of width at most $2k+1$. The running time of the algorithm is $O(dk^2)$ and the number of bags in the resulting path decomposition $\cC$ is $m\leq kd$. This solves an open problem stated in several papers, e.g. in \cite{connected_weighted_trees,searching_not_jumping,price_of_connectedness,guaranteed_graph_searching,connected_treewidth,monotony_properties_connected_visible,sweeping_large_cliques}, since it implies that for any graph $G$, $\cpw(G)\leq 2\pw(G)+1$, and improves previously known estimations \cite{price_of_connectedness,connected_searching_chordal_graphs}. The path decomposition $\cC$ computed by the algorithm can be used to obtain a monotone connected search strategy using at most $2k+3$ searchers. Thus, in terms of the graph searching terminology, the above bound immediately implies that $\mcsn(G)\leq\cpw(G)+2\leq 2\pw(G)+3\leq 2\sn(G)+3$. Since $\csn(G)\leq\mcsn(G)$, the bound can be restated for the connected search number of a graph, $\csn(G)\leq 2\sn(G)+3$. Moreover, the factor $2$ in the bound is tight~\cite{searching_not_jumping}. The bound can also be used to design approximation algorithms, for it implies that the pathwidth and the connected pathwidth (the search number, the connected search number, and some other search numbers not mentioned here, e.g. the internal search number) are within a constant factor of each other. We can also use the algorithm to construct a (monotone) connected search strategy for $2k+3$ searchers given, besides of $G$ and $\cP$, an input homebase vertex.

\section{Preliminaries and basic definitions}
\label{sec:definitions}

Given a simple graph $G=(V(G),E(G))$ and its subset of vertices $X\subseteq V(G)$, the subgraph of $G$ \emph{induced} by $X$ is
\[G[X]=(X,\{\{u,v\}\in E(G)\colon u,v\in X\}).\]
For a simple (not necessary connected) graph $G$, $H$ is a connected component of $G$ if $H$ is connected, that is, there exists a path in $H$ between each pair of vertices, and each proper supergraph of $H$ is not a subgraph of $G$. For $X\subseteq V(G)$ let
\[N_G(X)=\{u\in V(G)\setminus X\colon\{u,x\}\in E(G)\textup{ for some }x\in X\}.\]

\begin{defn}
A \emph{path decomposition} of a simple graph $G=(V(G),E(G))$ is a sequence $\cP=(X_1,\ldots,X_d)$, where $X_i\subseteq V(G)$ for each $i=1,\ldots,d$, and
\begin{itemize}
 \item[$\circ$] $\bigcup_{i=1,\ldots,d}X_i=V(G)$,
 \item[$\circ$] for each $\{u,v\}\in E(G)$ there exists $i\in\{1,\ldots,d\}$ such that $u,v\in X_i$,
 \item[$\circ$] for each $i,j,k$, $1\leq i\leq j\leq k\leq d$ it holds $X_i\cap X_k\subseteq X_j$.
\end{itemize}
The \emph{width} of the path decomposition $\cP$ is $\width(\cP)=\max_{i=1,\ldots,d}|X_i|-1$. The \emph{pathwidth} of $G$, $\pw(G)$, is the minimum width over all path decompositions of $G$.
\end{defn}

A path decomposition $\cP$ is \emph{connected} if $G[X_1\cup\cdots\cup X_i]$ is connected for each $i=1,\ldots,d$. We use the symbol $\cpw(G)$ to denote the minimum width over all connected path decompositions of $G$.

\begin{defn}
Given a graph $G$ and its path decomposition $\cP=(X_1,\ldots,X_d)$, a node-weighted graph $\cG=(V(\cG),E(\cG),\omega)$ \emph{derived} from $G$ and $\cP$ is the graph with vertex set
\[V(\cG)=V_1\cup\cdots\cup V_d,\]
where $V_i=\{v_i(H)\colon H\textup{ is a connected component of }G[X_i]\}$, $i=1,\ldots,d$, and edge set
\begin{multline} \nonumber
E(\cG)=\{\{v_i(H),v_{i+1}(H')\}\colon v_i(H)\in V_i,v_{i+1}(H')\in V_{i+1},
i\in\{1,\ldots,d-1\},\textup{ and }V(H)\cap V(H')\neq\emptyset\}.
\end{multline}
The weight of a vertex $v_i(H)\in V(\cG)$, $i\in\{1,\ldots,d\}$, is $\omega(v_i(H))=|V(H)|$. The \emph{width} of $\cG$, denoted by $\width(\cG)$, equals $\width(\cP)+1$.
\end{defn}
In the following we omit a subgraph $H$ of $G$ and the index $i\in\{1,\ldots,d\}$ whenever they are not important when referring to a vertex of $\cG$ and we write $v$ instead of $v_i(H)$. For brevity, $\omega(X)=\sum_{x\in X}\omega(x)$ for any subset $X\subseteq V(\cG)$.

Figures~\ref{fig:decomposition_ex}(a) and~\ref{fig:decomposition_ex}(b) present a graph $G$ and its path decomposition $\cP$, respectively, where the subgraph structure in each bag $X_i$ is also given. Figure~\ref{fig:decomposition_ex}(c) depicts the derived graph $\cG$. Note that $\cP$ is not connected: the subgraphs $G[X_1\cup\cdots\cup X_i]$ are not connected for $i=2,3,4$.
\begin{figure}[ht]
	\begin{center}
	\includegraphics{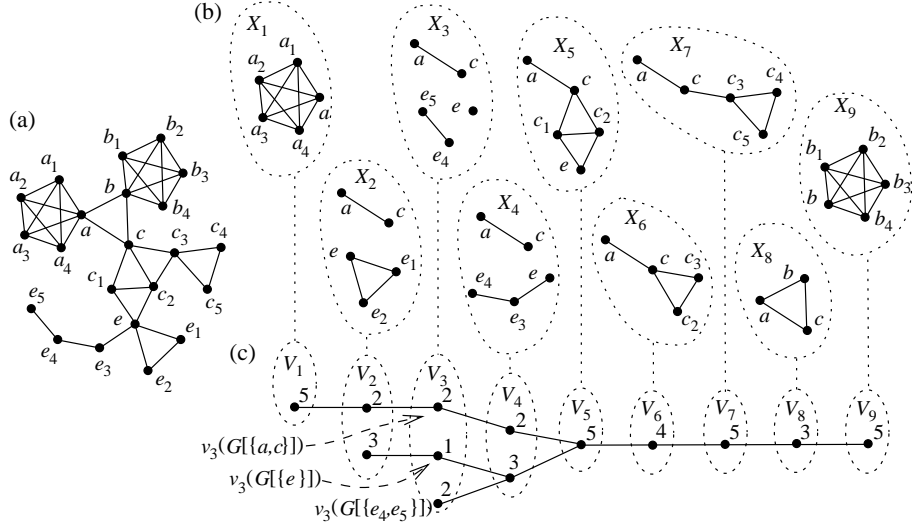}
	\caption{(a) a graph $G$; (b) a path decomposition $\cP$ of $G$; (c) the weighted graph $\cG$ derived from $G$ and $\cP$}
	\label{fig:decomposition_ex}
	\end{center}
\end{figure}

Let $C\subseteq V(\cG)$. The \emph{border} $\border(C)$ of the set $C$ is its subset consisting of all the vertices $v\in C$ such that there exists $u\in V(\cG)\setminus C$ adjacent to $v$ in $\cG$, i.e. $\border(C)=N_{\cG}(V(\cG)\setminus C)$. The algorithms presented in Sections~\ref{sec:simple_algorithm} and~\ref{sec:algorithm} maintain, besides a set $C$ and its border $\border(C)$, a partition $\lborder(C),\rborder(C)$ such that $\border(C)=\lborder(C)\cup\rborder(C)$ and
\begin{equation} \label{eq:border_partition}
\lborder(C)\subseteq V_1\cup\cdots\cup V_i,\quad \rborder(C)\subseteq V_{i+1}\cup\cdots\cup V_d,
\end{equation}
for some $0\leq i\leq d$. We prove that the partitions used by the algorithms do satisfy the above condition and for the time being we continue with the assumption that for each $C$ and $\border(C)$ such a partition is given.

Given a set $X\subseteq V(\cG)$, $X\neq\emptyset$, we define the \emph{left} (\emph{right}) \emph{extremity} of $X$ as
$l(X)=\min\{i\colon V_i\cap X\neq\emptyset\}$ ($r(X)=\max\{i\colon V_i\cap X\neq\emptyset\}$, respectively). Note that (\ref{eq:border_partition}) in particular implies $r(\lborder(C))<l(\rborder(C))$.

\section{A simple conversion algorithm} \label{sec:simple_algorithm}

In this section we present a simple algorithm that takes $G$ and a path decomposition $\cP$ of $G$ as an input, and returns a connected path decomposition $\cC$ of $G$. However, it is not guaranteed that $\width(\cC)/\width(\cP)$ is bounded by a constant. On the other hand, both this one and our final algorithm presented in Section~\ref{sec:algorithm} can be seen as the sequences of executions of the basic steps, and the difference lies in using different criteria while constricting those sequences.

The first computation performed by $\AlgSimpleConnectedPathwidth$ (\emph{Simple Connected Pathwidth}) is the construction of the derived graph $\cG$ and in the subsequent steps the algorithm works on $\cG$. (Also, most parts of our analysis use $\cG$ rather than $G$.) The algorithm computes a sequence of sets $C_j\subseteq V(\cG)$, $j=1,\ldots,m$, called \emph{expansions}. In addition to that, the sets $A_j\subseteq V(\cG)$, $j=2,\ldots,m$, and $B_j\subseteq V(\cG)$, $j=1,\ldots,m$, are computed. The former one consists of the vertices that are added to $C_{j-1}$ to obtain $C_j$, while $B_j$ is used to determine the vertices of $G$ that belong to the $j$-th bag of the resulting path decomposition $\cC$.
Informally speaking, $A_j$ consists of some vertices in $N_{\cG}(C_{j-1})$, and $B_j=\border(C_j)\cup A_j$, $j=2,\ldots,m$. The expansion $C_1$ consists of any vertex in $V_1$, and $C_m=V(\cG)$ at the end of the execution of $\AlgSimpleConnectedPathwidth$. Moreover, $C_j\varsubsetneq C_{j+1}$ for each $j=1,\ldots,m-1$. This guarantees that the final path decomposition obtained from $B_1,\ldots,B_m$ is valid and connected, as proven in Lemma~\ref{lem:scp_correctness}.
By construction, $\omega(B_j)$ is the size of the corresponding $j$-th bag of $\cC$, but we do not attempt to bound the width of the path decomposition returned by $\AlgSimpleConnectedPathwidth$. In this section, besides of the statement of the algorithm, we prove its correctness, i.e. that it stops and returns a connected path decomposition.

The algorithm computes for each expansion $C_j$ two disjoint sets called the \emph{left} and \emph{right borders} of $C_j$ (introduced informally in Section~\ref{sec:definitions}), denoted by $\lborder(C_j)$ and $\rborder(C_j)$, respectively. It is guaranteed that $\lborder(C_j)\cup\rborder(C_j)=\border(C_j)$ for each $j=1,\ldots,m$. As it is proven later, the left and right borders are special types of partitions of $\border(C_j)$. In particular, for each $j=1,\ldots,m$ there exists an integer $i\in\{0,\ldots,d\}$ such that the left border $\lborder(C_j)$ is contained in $V_1\cup\cdots V_i$ and the right border $\rborder(C_j)$ is a subset of $V_{i+1}\cup\cdots\cup V_d$. For brevity let $l(\lborder(C_j))=r(\lborder(C_j))=0$ if $\lborder(C_j)=\emptyset$ and $l(\rborder(C_j))=r(\rborder(C_j))=d+1$ if $\rborder(C_j)=\emptyset$, where $C_j$ is any expansion. The above-mentioned properties are not needed while proving the correctness of the algorithms, but are necessary for proving the upper bound on $\width(\cC)$, where $\cC$ is returned by the algorithm from Section~\ref{sec:algorithm}. For this reason their formal statements and proofs are postponed till Section~\ref{sec:bound}.

As mentioned earlier, both of our algorithms use two basic steps, called $\EL$ (\emph{Left Extension}) and $\ER$ (\emph{Right Extension}). We describe them first and then we give the pseudo-code of $\AlgSimpleConnectedPathwidth$. The steps use $m$, $A_m,B_m,C_m$, the borders $\lborder(C_m)$ and $\rborder(C_m)$, and the derived graph $\cG$ as global variables. Both steps are symmetric, they increment $m$ and compute the above list of sets for the new index $m$. The input is an integer $i\in\{1,\ldots,d\}$. Informally speaking, the new set $C_m$ is computed by adding to $C_{m-1}$ the selected vertices among those in $V_{i-1}$ in case of Step~$\EL$ and in $V_{i+1}$ in case of $\ER$ that are adjacent to the vertices in $C_{m-1}$.

\begin{center}
\begin{tabular}{rcl}
\parbox{0.45\textwidth}{
\algrule
\f{0}{\textbf{Step} $\EL$ (\emph{Left Extension})}
\algrule
\f[10pt]{0}{\textbf{Input:} An integer $i\in\{1,\ldots,d\}$. ($\cG$, $m$, $A_m,B_m,C_m$, and the borders are used as global variables).}
\f{0}{\textbf{begin}}
\f{1}{   Increment $m$.}
\f{1}{   $A_m:=V_{i-1}\cap(N_{\cG}(\border(C_{m-1})\cap V_{i})\setminus C_{m-1})$,}
\f{1}{   $C_m:=C_{m-1}\cup A_m$,}
\f{1}{   $B_m:=\border(C_m)\cup A_m$,}
\f{1}{   $\lborder(C_m):=(\lborder(C_{m-1})\cup A_m)\cap\border(C_m)$,}
\f{1}{   $\rborder(C_m):=\rborder(C_{m-1})\cap\border(C_m)$,}
\f{0}{\textbf{end Step} $\EL$.}
\algrule
}
&\hspace*{0.02\textwidth}&
\parbox{0.45\textwidth}{
\algrule
\f{0}{\textbf{Step} $\ER$ (\emph{Right Extension})} \algrule
\f[10pt]{0}{\textbf{Input:} An integer $i\in\{1,\ldots,d\}$. ($\cG$, $m$, $A_m,B_m,C_m$, and the borders are used as global variables).}
\f{0}{\textbf{begin}}
\f{1}{   Increment $m$.}
\f{1}{   $A_m:=V_{i+1}\cap(N_{\cG}(\border(C_{m-1})\cap V_{i})\setminus C_{m-1})$,}
\f{1}{   $C_m:=C_{m-1}\cup A_m$,}
\f{1}{   $B_m:=\border(C_m)\cup A_m$,}
\f{1}{   $\rborder(C_m):=(\rborder(C_{m-1})\cup A_m)\cap\border(C_m)$,}
\f{1}{   $\lborder(C_m):=\lborder(C_{m-1})\cap\border(C_m)$,}
\f{0}{\textbf{end Step} $\ER$.}
\algrule
}
\end{tabular}
\end{center}

Figure~\ref{fig:scp_example}(a) depicts an exemplary subgraph $\cG$ (induced by $V_1\cup\cdots\cup V_5$) together with white vertices in an expansion $C_m$. In all cases (including the following figures) $\diamondsuit$ and $\square$ are used to denote the vertices of the right and left borders, respectively. Figures~\ref{fig:scp_example}(b) and~\ref{fig:scp_example}(c) show the result of the execution of $\EL(r(\lborder(C_m)))$ and $\ER(r(\lborder(C_m)))$, respectively, i.e. the corresponding expansion $C_{m+1}$. (Note that $r(\lborder(C_m))=3$ in this example.) We always use $\EL$ and $\ER$ with input $i\in\{r(\lborder(C_m)),l(\rborder(C_m))\}$.
\begin{figure}[ht]
	\begin{center}
	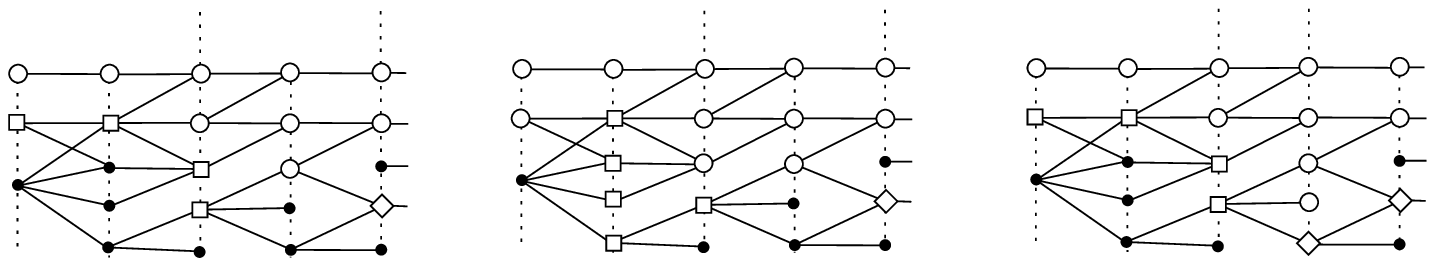
	\caption{(a) an expansion $C_m$; (b) $C_{m+1}$ after the execution of $\EL(r(\lborder(C_m)))$; (c) $C_{m+1}$ after the execution of $\ER(r(\lborder(C_m)))$}
	\label{fig:scp_example}
	\end{center}
\end{figure}

\algrule
\f{0}{\textbf{Procedure} $\AlgSimpleConnectedPathwidth$ (\emph{Simple Connected Pathwidth})} \algrule
\f{0}{\textbf{Input:} A simple graph $G$ and a path decomposition $\cP$ of $G$.}
\f{0}{\textbf{Output:} A connected path decomposition $\cC$ of $G$.}
\f{0}{\textbf{begin}}
\f{1}{ Use $G$ and $\cP$ to calculate the derived graph $\cG$. Let $v$ be any vertex in $V_1$ and let $B_1=C_1=\{v\}$ and $\lborder(C_1)=\emptyset$, $\rborder(C_1)=\{v\}$. Let $m=1$. If $v$ has no neighbors in $\cG$, then return $\cP$.}
\f{1}{\textbf{while} $C_m\neq V(\cG)$ \textbf{do}}
\f{2}{   Execute any of the following Steps~S1-S4 that result in computing $C_m$ such that $C_{m}\neq C_{m-1}$:}
\f{3}{        S1: Call $\EL(r(\lborder(C_m)))$.}
\f{3}{        S2: Call $\ER(l(\rborder(C_m)))$.}
\f{3}{        S3: Call $\ER(r(\lborder(C_m)))$.}
\f{3}{        S4: Call $\EL(l(\rborder(C_m)))$.}
\f{1}{\textbf{end while}}
\f{1}{Let $Z_j=\bigcup_{v(H)\in B_j}V(H)$ for each $j=1,\ldots,m$. \textbf{Return} $\cC=(Z_1,\ldots,Z_m)$.}
\textbf{end procedure} $\AlgSimpleConnectedPathwidth$.
\algrule\\

First we briefly discuss the initialization stage of $\AlgConnectedPathwidth$. The first expansion $C_1$ consists of a single vertex $v\in V_1$. If $v$ has no neighbors in $\cG$, then due to the connectedness of $G$, $\cP$ consists of a single bag, and therefore $\cP$ is connected itself.

The instructions prior to the while loop of $\AlgSimpleConnectedPathwidth$ compute the first set of variables, i.e. for $m=1$. In the following, one \emph{iteration} of $\AlgSimpleConnectedPathwidth$ means one iteration of its `while' loop, which reduces to the execution of one of Steps S1-S4. We discuss Steps~S1 and~S3, because the other ones are symmetric (S2 is symmetric to S1 and S3 is symmetric to S4). Let us consider Step~S1. (We refer here to $C_m$ from the beginning of the execution of Step~$\EL$). If no vertex in $V_{r(\lborder(C_m))}\cap\lborder(C_m)$ has a neighbor in $V_{r(\lborder(C_m))+1}\setminus C_m$, then the execution of Step~$\EL$ guarantees that the right extremity of the left border that will be obtained in $\EL$ is strictly less than $r(\lborder(C_m))$, i.e. $r(\lborder(C_{m+1}))<r(\lborder(C_m))$. The right border of the new expansion computed by $\EL$ in Step~S1 always is equal to the right border of the previous expansion. Note that if no vertex in $V_{r(\lborder(C_m))}\cap\lborder(C_m)$ has a neighbor in $V_{r(\lborder(C_m))-1}\setminus C_m$, then $\EL$ called in Step~S1 would compute $C_{m+1}$ that is equal to $C_{m}$ and therefore Step~S1 is not executed in such case. Step~S3, on the other hand, guarantees that the left border of the new expansion is contained in the left border of the previous one. Then, informally speaking, the right border of the new expansion, that is $\rborder(C_{m+1})$, besides some vertices from $\rborder(C_m)$, consists of some vertices in $V_{r(\lborder(C_m))+1}$.

Now we give one preliminary lemma and then we prove that $\cC$ returned by $\AlgSimpleConnectedPathwidth$ is a connected path decomposition of $G$.

\begin{lemma} \label{lem:C_is_connected}
$\cG[C_j]$ is connected for each $j=1,\ldots,m$.
\end{lemma}
\begin{proof}
By induction on $j\in\{1,\ldots,m\}$. $C_1$ is connected, because it consists of a single vertex of $\cG$. Suppose that $C_j$, $1\leq j<m$, is connected and let us consider $C_{j+1}$. The computation of the latter one is performed in Step~$\EL$ or Step~$\ER$. By the definition, $C_{j+1}=C_j\cup A_{j+1}$ for some $i\in\{1,\ldots,d\}$. The connectedness of $\cG[C_{j+1}]$ follows from the definition of $N_{\cG}$ and from the fact that $A_{j+1}\subseteq N_{\cG}(C_j)$.
\end{proof}

\begin{lemma} \label{lem:scp_correctness}
Given a simple graph $G$ and its path decomposition $\cP=(X_1,\ldots,X_d)$, $\AlgSimpleConnectedPathwidth$ returns a connected path decomposition $\cC=(Z_1,\ldots,Z_m)$ of $G$.
\end{lemma}
\begin{proof}
Note that $C_m=V(\cG)$. This follows from an observation that in each iteration of $\AlgSimpleConnectedPathwidth$ at least one of Steps~S1-S4 guarantees to compute $C_j$ such that $C_j\neq C_{j-1}$. Indeed, if $\lborder(C_{j-1})\neq\emptyset$, then a vertex $x\in V_i\cap\lborder(C_{j-1})$, $i=r(\lborder(C_{j-1}))$, has a neighbor $v$ in either $V_{i-1}\setminus C_{j-1}$ or in $V_{i+1}\setminus C_{j-1}$. By the formulation of Steps~$\EL$ and~$\ER$ we obtain that $v\in C_j$ if Step~S1 or Step~S3, respectively, are performed. An analogous argument holds for the right border. Note that $\lborder(C_j)=\rborder(C_j)=\emptyset$ is, by the connectedness of $G$, equivalent to $C_j=V(\cG)$. Thus, the execution of $\AlgSimpleConnectedPathwidth$ stops and the algorithm returns $\cC=(Z_1,\ldots,Z_m)$.

Now we prove that $\cC$ is a path decomposition of $G$.
Let $u$ be any vertex of $G$. Since $\cP$ is a path decomposition, $u\in X_i$ for some $i\in\{1,\ldots,d\}$. Therefore, $u$ is a vertex of a subgraph $H$ such that $v_i(H)\in V_i$. Since $C_m=V(\cG)$ and $C_j\subseteq C_{j+1}$ for each $j=1,\ldots,m-1$, we obtain that there exists a minimum integer $j\in\{1,\ldots,m\}$ such that $v_i(H)\in C_j$. By construction, $v_i(H)\in A_j$ and therefore $v_i(H)\in B_j$. Thus, $u\in Z_j$.

Similarly, for each $\{u,v\}\in E(G)$ there exists $j\in\{1,\ldots,m\}$ such that $u,v\in Z_j$. Indeed, $u,v\in X_i$ for some $i\in\{1,\ldots,d\}$ implies, as before, that $\{u,v\}$ is an edge of some subgraph $H$ of $G$ such that $v_i(H)\in B_j$.

Let $i,k$ be integers, $1\leq i\leq k\leq m$, and suppose that $u\in Z_i\cap Z_k$. We show that $u\in Z_j$ for each $i\leq j\leq k$. First note that the subgraph $\cG[U]$, where $U=\{v_i(H)\colon u\in V(H),i=1,\ldots,d\}$, is connected, because, if $u\in V(H)$, $v_s(H)\in V_{s}$ and $u\in V(H')$, $v_{s'}(H)\in V_{s'}$, $s\leq s'$, then by the definition of the derived graph $u\in X_s$ and $u\in X_{s'}$. Since $\cP$ is a path decomposition, $u\in X_p$ for each $p=s,\ldots,s'$. Therefore, for each $p=s,\ldots,s'$, there exists a vertex $v_p(H_p'')$ in $V_p$ such that $u\in V(H_p'')$. By the definition, $v_p(H_p'')\in U$. Since $C_j\subseteq C_{j+1}$ and the Steps~$\EL$ and~$\ER$ compute $C_{j+1}$ by taking the vertices in $C_j$ and the subset of $N_{\cG}(C_j)$, $j=1,\ldots,m-1$, the connectedness of $\cG[U]$ gives us that $u\in Z_j$ for $i\leq j\leq k$.

Finally we prove that $G[Z_1\cup\cdots\cup Z_j]$ is connected for each $j=1,\ldots,m$. Let $u,u'\in Z_1\cup\cdots\cup Z_j$. From the formulation of Steps~$\EL$ and~$\ER$ it follows that $C_j=B_1\cup\cdots\cup B_j$. Thus, $C_j$ consists of the vertices $v_i(H)$ such that $H$ is a connected component of $G[Z_p]$, where $p\in\{1,\ldots,j\}$. By the definition of derived graph, there exist two vertices $v_i(H)$, $v_{i'}(H')$ of $\cG$ such that $u\in V(H)$, $u'\in V(H')$ and $v_i(H)\in C_j$, $v_{i'}(H')\in C_j$. By Lemma~\ref{lem:C_is_connected}, there exists a path $P$ in $\cG[C_j]$ connecting $v_i(H)$ and $v_{i'}(H')$. Let the consecutive vertices of $P$ be $v_{i}(H)=v_{i_1}(H_1),\ldots,v_{i_p}(H_{p})=v_{i'}(H')$. By the definition of $\cG$, $H_{s}$ and $H_{s+1}$ share a vertex of $G$, $s=1,\ldots,p-1$. Moreover, $H_{s}$ is connected for each $s=1,\ldots,p$. This implies that there exists a path $P'$ in $G$ between each vertex of $H$ and each vertex of $H'$, in particular between $u$ and $u'$, and $V(P')\subseteq\bigcup_{1\leq s\leq p}V(H_{s})$. Since $v_i(H)\in C_j$ implies $v_i(H)\in B_s$ for some $s\leq j$, and consequently, $V(H)\subseteq Z_{s}$, we obtain that $V(P')\subseteq Z_1\cup\cdots\cup Z_j$, which proves the connectedness of $\cC$.
\end{proof}

\section{The branches} \label{sec:branches}

In this section we introduce the concept of \emph{branches}, our main tool for organizing the sequences of consecutive executions of Steps~$\EL$ and~$\ER$.
A path $P$ in $\cG$ is \emph{progressive} if $|V(P)\cap V_i|\leq 1$ for each $i=1,\ldots,d$. Note that a progressive path that connects a vertex in $V_i$ to a vertex in $V_j$ consists of exactly $|i-j|$ edges, or in other words, a progressive path contains exactly one vertex in $V_s$ for each $s=\min\{i,j\},\ldots,\max\{i,j\}$.

\begin{defn}
Given $\cG$ and $C\subseteq V(\cG)$, a \emph{left} \emph{branch} $\lbranch(C,i)$, where $i\leq r(\lborder(C))$ is the subgraph of $\cG$ induced by the vertices in $\lborder(C)$ and by the vertices $v\in V_k\setminus C$, where $i\leq k<r(\lborder(C))$, connected by a progressive path to a vertex $x\in V_j\cap\lborder(C)$ for some $k<j$.

A \emph{right} \emph{branch} $\rbranch(C,i)$, where $i\geq l(\rborder(C))$ is the subgraph of $\cG$ induced by the vertices in $\rborder(C)$ and by the vertices $v\in V_k\setminus C$, where $l(\rborder(C))<k\leq i$, connected by a progressive path to a vertex $x\in V_j\cap\rborder(C)$ for some $j<k$.
\end{defn}
Informally speaking, we construct $\lbranch(C,i)$ by taking $\lborder(C)$ and all vertices in $V_i\cup\cdots\cup V_{r(\lborder(C))}$ achievable from the vertices $u\in\lborder(C)$ with progressive paths having $u$ as the right endpoint.
We sometimes write $\branch$ to refer to a branch whenever its `direction' or $C,i$ are clear from the context.

Let $\branch=\lbranch(C,i)$, $i\leq r(\lborder(C))$, ($\branch=\rbranch(C,i)$, $i\geq l(\rborder(C))$) be a branch.
A vertex $v$ of $\branch$ is \emph{external} if $N_{\cG}(v)\nsubseteq C\cup V(\branch)$.
The branch $\branch$ is \emph{proper} if it has no external vertices in $V_j$ for each $j>i$ ($j<i$, respectively), while $\branch$ is \emph{maximal} if it has no external vertices or if it is proper and $\lbranch(C,i-1)$ ($\rbranch(C,i+1)$, respectively) is not a proper branch. Let $\Ext(\branch)$ denote the set of all external vertices of $\branch$.

Informally speaking, the external vertices of a branch $\branch$ are the ones that have neighbors not in $C\cup V(\branch)$. A left branch $\lbranch(C,i)$ is proper if we can `grow' it from $r(\lborder(C))$ to $i$ without leaving any external vertices in $V_j$, $j>i$. Then, the maximality of the branch implies that we cannot grow the branch beyond $i$, because either $i=1$ or the new branch would have an external vertex in $V_i$ (in such case this external vertex must have a neighbor in $V_{i+1}\setminus(C\cup V(\branch))$).

Figure~\ref{fig:branch} illustrates the above definitions. (In all cases the branch is distinguished by the dark area.)
\begin{figure}[ht]
	\begin{center}
	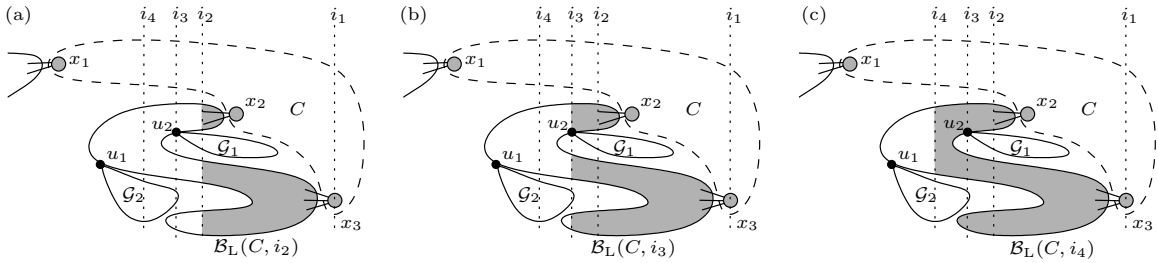
	\caption{$\cG$ with distinguished vertex sets $C$ and $\lborder(C)=\{x_1,x_2,x_3\}$, and the corresponding (left) branches that are:
           (a) proper but not maximal;
           (b) maximal;
           (c) not proper (thus not maximal)}
	\label{fig:branch}
	\end{center}
\end{figure}
Let $\lborder(C)=\{x_1,x_2,x_3\}$. Figure~\ref{fig:branch}(a) gives $\lbranch(C,i_2)$ and this branch is proper, but not maximal for each $i_2$, $i_3<i_2\leq i_1$. The lack of maximality is due to the fact that, informally speaking, we can `grow' $\lbranch(C,i_2)$ by including the corresponding vertices in $V_{i_2-1}$ and the new branch is still proper, as none of its vertices in $V_{i_2}$ are external. The branch $\lbranch(C,i_3)$ (see Figure~\ref{fig:branch}(b)) is maximal (thus proper), which follows from the fact that any branch $\lbranch(C,i_4)$, where $i_4<i_3$, is not proper, because it contains an external vertex $u_2$, as shown in Figure~\ref{fig:branch}(c). Note that the vertices of $\cG_1$ and $\cG_2$ (except for $u_1$ and $u_2$) do not belong to any left branch $\lbranch(C,i)$, because they are not connected by progressive paths to $x_2$ or $x_3$. In our algorithm we ensure that each branch we use is proper.

An integer $j$ is a \emph{cut} of a left branch $\branch=\lbranch(C,i)$ if $j\in\{i,\ldots,r(V(\branch))\}$. Then, its \emph{weight} equals $\omega(\Ext(\lbranch(C,j)))$.
An integer $j$ is a \emph{cut} of a right branch $\branch=\rbranch(C,i)$ if $j\in\{l(V(\branch)),\ldots,i\}$. Then, its \emph{weight} equals $\omega(\Ext(\rbranch(C,j)))$.
A cut of minimum weight is a \emph{bottleneck} of a branch.

We finish this section with the following observations.
\begin{observation} \label{obs:o1}
For each expansion $C$ it holds $V_j\cap V(\lbranch(C,i))=V_j\cap V(\lbranch(C,i'))$ for each $i\leq i'\leq j\leq r(\lborder(C))$ and $V_j\cap V(\rbranch(C,i))=V_j\cap V(\rbranch(C,i'))$ for each $l(\rborder(C))\leq j\leq i'\leq i$.
\qed
\end{observation}
\begin{observation} \label{obs:o2}
For each expansion $C$ the weight of a cut $j$ of a proper branch $\lbranch(C,i)$, $i\leq j\leq r(\lborder(C))$, is less than or equal to $\omega((V_1\cup\cdots\cup V_{j-1})\cap\lborder(C))+\omega(V_j\cap V(\lbranch(C,i)))$, and the weight of a cut $j$ of a proper branch $\rbranch(C,i)$, $l(\rborder(C))\leq j\leq i$, is less than or equal to $\omega((V_{j+1}\cup\cdots\cup V_{d})\cap\rborder(C))+\omega(V_j\cap V(\rbranch(C,i)))$.
\qed
\end{observation}

\section{The algorithm} \label{sec:algorithm}

We start with two subroutines $\AlgProcessLeftBranch$ (\emph{Process Left Branch}) and $\AlgProcessRightBranch$ (\emph{Process Right Branch}) that are used by the main procedure given in this section. The input to $\AlgProcessLeftBranch$ and to $\AlgProcessRightBranch$ consists of an integer $t$, $t\in\{1,\ldots,d\}$. In the following we say for brevity that an expansion has been computed by $\AlgProcessLeftBranch$ or $\AlgProcessRightBranch$ whenever it has been computed by $\EL$ or $\ER$ called by $\AlgProcessLeftBranch$ or $\AlgProcessRightBranch$, respectively.  Due to symmetry we skip the informal description of $\AlgProcessRightBranch$ here. If $C_{j+1}$ and $C_{j'}$ are the first and the last expansions computed by $\AlgProcessLeftBranch$, then $C_{j'}=C_j\cup V(\lbranch(C_j,t))$, which we formally show in Lemma~\ref{lem:branches}. This is achieved by several executions of Step~$\EL$. In particular, $\EL(r(\lborder(C_m)))$ is repeatedly called as long as the right extremity of the left border of the current expansion is greater than $t$. The procedures $\AlgProcessLeftBranch(t)$ and $\AlgProcessRightBranch(t)$ are as follows.

\noindent
\begin{center}
\begin{tabular}{rcl}
\parbox{0.45\textwidth}{
\algrule
\f{0}{\textbf{Procedure} $\AlgProcessLeftBranch$ (\emph{Process Left Branch})} \algrule
\f[10pt]{0}{\textbf{Input:} An integer $t$. ($m$, $C_m$ and $\lborder(C_m)$ are used as global variables)}
\f{0}{\textbf{begin}}
\f{1}{      \textbf{while} $r(\lborder(C_m))>t$ \textbf{do}}
\f{2}{          Call $\EL(r(\lborder(C_m)))$.}
\f{0}{\textbf{end procedure} $\AlgProcessLeftBranch$.}
\algrule
}
&\hspace*{0.03\textwidth}&
\parbox{0.45\textwidth}{
\algrule
\f{0}{\textbf{Procedure} $\AlgProcessRightBranch$ (\emph{Process Right Branch})} \algrule
\f[10pt]{0}{\textbf{Input:} An integer $t$. ($m$, $C_m$ and $\rborder(C_m)$ are used as global variables)}
\f{0}{\textbf{begin}}
\f{1}{      \textbf{while} $l(\rborder(C_m))<t$ \textbf{do}}
\f{2}{          Call $\EL(l(\rborder(C_m)))$.}
\f{0}{\textbf{end procedure} $\AlgProcessRightBranch$.}
\algrule
}
\end{tabular}
\end{center}

Now we are ready to give the pseudo-code of the main algorithm $\AlgConnectedPathwidth$ (\emph{Connected Pathwidth}). Its input consists of, as in the case of $\AlgSimpleConnectedPathwidth$, a simple graph $G$ and its path decomposition $\cP$.

\noindent
\algrule
\f{0}{\textbf{Algorithm} $\AlgConnectedPathwidth$ (\emph{Connected Pathwidth})}\algrule
\f{0}{\textbf{Input:} A simple graph $G$ and a path decomposition $\cP$ of $G$.}
\f{0}{\textbf{Output:} A connected path decomposition $\cC$ of $G$.}
\f{0}{\textbf{begin}}
\f{1}{  (\emph{Initialization}.)}
\f[20pt]{2}{    I.1: Use $G$ and $\cP$ to calculate the derived graph $\cG$.  Let $v$ be any vertex in $V_1$ and let $B_1=C_1=\{v\}$ and $\lborder(C_1)=\emptyset$, $\rborder(C_1)=\{v\}$. Let $m=1$. If $v$ has no neighbors in $\cG$, then return $\cP$.}
\f[20pt]{2}{    I.2: Find the maximal right branch $\rbranch(C_1,a_0)$ with a bottleneck $a_0'$ ($1\leq a_0'\leq a_0$). Call $\AlgProcessRightBranch(a_0')$.}
\f{1}{  (\emph{Main loop}.)}
\f{2}{    \textbf{while} $C_m\neq V(\cG)$ \textbf{do}}
\f{3}{      \textbf{if} $\omega(\lborder(C_m))>\omega(\rborder(C_m))$ \textbf{then}}
\f{4}{         L.1: Find the maximal left branch $\branch_1=\lbranch(C_m,t_1)$. Call $\AlgProcessLeftBranch(t_1)$.}
\f{4}{         L.2: Call $\ER(t_1)$.}
\f{3}{      \textbf{else}}
\f{4}{         R.1: Find the maximal right branch $\branch_1=\rbranch(C_m,t_1)$. Call $\AlgProcessRightBranch(t_1)$.}
\f{4}{         R.2: Call $\EL(t_1)$.}
\f{3}{      \textbf{end if}}
\f[20pt]{3}{    RL.3: Find the maximal right and left branches $\branch_2=\rbranch(C_m,t_2)$ and $\branch_3=\lbranch(C_m,t_3)$, with bottlenecks $t_2'$ and $t_3'$, respectively. Call $\AlgProcessLeftBranch(t_3')$ and $\AlgProcessRightBranch(t_2')$.}
\f{2}{    \textbf{end while.}}
\f{2}{    Let $Z_j=\bigcup_{v(H)\in B_j}V(H)$ for each $j=1,\ldots,m$. \textbf{Return} $\cC=(Z_1,\ldots,Z_m)$.}
\f{0}{\textbf{end procedure} $\AlgConnectedPathwidth$.}
\algrule\\

Note that Step~I.1 of $\AlgConnectedPathwidth$ consists of the instructions from the initialization stage of $\AlgSimpleConnectedPathwidth$.
Then, the algorithm finds in Step~I.2 the maximal right branch `emanating' from $v$, and includes into $C_1$ the part of the branch that `reaches' a bottleneck of the branch. This is achieved by the call to $\AlgProcessRightBranch$ with the input $a_0'$.

In the following, one \emph{iteration} of $\AlgConnectedPathwidth$, $\AlgProcessLeftBranch$ or $\AlgProcessRightBranch$ means one iteration of the `while' loop in the corresponding procedure. Thus, in the case of $\AlgConnectedPathwidth$, one iteration reduces to executing Steps~L.1, L.2, RL.3 or R.1, R.2, RL.3, while in the procedures $\AlgProcessLeftBranch$ and $\AlgProcessRightBranch$ one iteration results in executing Step~$\EL$ or Step~$\ER$, respectively. We use the symbols $\branch_i,t_i$, $i=1,2,3$, and $t_i'$, $i=2,3$, to refer to the variables used in $\AlgConnectedPathwidth$. In what follows we denote for brevity $\branch_2'=\rbranch(C_m,t_2')$ and $\branch_3'=\lbranch(C_m,t_3')$. Informally speaking, $\branch_2'$ and $\branch_3'$ are the branches $\branch_2$ and $\branch_3$, respectively, restricted to the vertices up to the corresponding cut $t_2'$ or $t_3'$. Note that the outcome in Step~RL.3 (in terms of the expansion obtained at the end of the iteration) is the same regardless of the order of making the calls to $\AlgProcessLeftBranch$ and $\AlgProcessRightBranch$. Due to the analysis in Section~\ref{sec:bound}, the approximation guarantee of $\AlgConnectedPathwidth$ remains the same for each order of making those calls in Step~RL.3.

The branches are used in the subsequent iterations of $\AlgConnectedPathwidth$ in the way presented in Figure~\ref{fig:cases}, where the Steps~L.1, L2 and RL.3 are shown (the execution of Steps~R.1 and R.2 is symmetric with respect to Steps~L.1 and L.2).
\begin{figure}[ht]
	\begin{center}
	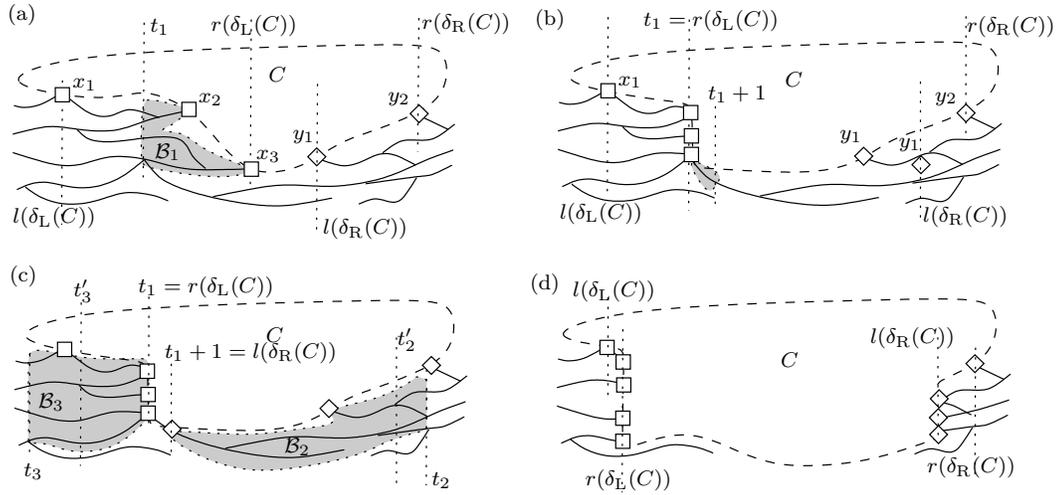
	\caption{The execution of one iteration of $\AlgConnectedPathwidth$ (Steps~L.1, L.2, RL.3):
                 (a) an expansion $C$ with $\lborder(C)=\{x_1,x_2,x_2\}$ and $\rborder(C)=\{y_1,y_2\}$ and the branch $\branch_1$ from Step~L.1;
                 (b) $C$ at the end of Step~L.1, dark area marks the vertices to be included in Step~L.2;
                 (c) $C$ at the end of Step~L.2, together with $\branch_2$ and $\branch_3$ computed in Step~RL.3;
                 (d) $C$ at the end of the iteration}
	\label{fig:cases}
	\end{center}
\end{figure}
Figure~\ref{fig:cases}(a) gives $C$ together with $\lborder(C)$ and $\rborder(C)$. The dark area is the maximal left branch $\branch_1$ from Step~L.1. Note that if $t_1=r(\lborder(C))$, where $C$ is the expansion from the beginning of an iteration, then no new expansion is computed during Step~L.1, and the algorithm proceeds to Step~L.2. Such a situation occurs if $\branch_1$ contains only the vertices in $\lborder(C)$ or, equivalently, if a vertex in $V_{r(\lborder(C))}\cap\lborder(C)$ has a neighbor in $V_{r(\lborder(C))+1}\setminus C$. The result of the execution of Step~L.1 is shown in Figure~\ref{fig:cases}(b) together with dark area marking the nodes in $V_{t_1+1}$ that will be included into the expansion in Step~L.2 --- see Figure~\ref{fig:cases}(c) for the result of the execution of Step~L.2 and for the maximal branches $\branch_2$ and $\branch_3$ (also, the exemplary integers $t_2'$ and $t_3'$ are shown to indicate the corresponding `sub-branches' $\branch_2'$ and $\branch_3'$ that reach the minimum cuts of $\branch_2$ and $\branch_3$, respectively). If none of the vertices included into $C$ in Step~L.2 are external, then no new right border vertices in $V_{t_1+1}$ are introduced, and therefore the new expansion in Figure~\ref{fig:cases}(c) has the right border that is a subset of the right border of the expansion in Figure~\ref{fig:cases}(b). Finally, Figure~\ref{fig:cases}(d) gives the result of the execution of RL.3 (and thus the entire iteration).

In this section we prove the correctness of $\AlgConnectedPathwidth$, while Section~\ref{sec:bound} analyzes the width of $\cC$ returned by $\AlgConnectedPathwidth$.
Lemmas~\ref{lem:one_branch} and~\ref{lem:branches} given below demonstrate how the expansions change between the subsequent calls to $\AlgProcessLeftBranch$ and $\AlgProcessRightBranch$. They show that, informally speaking, if (during the execution of $\AlgConnectedPathwidth$) we take a proper branch $\lbranch(C_m,i)$ or $\rbranch(C_m,i)$ and call $\AlgProcessLeftBranch(i)$ or $\AlgProcessRightBranch(i)$, respectively, then the expansion obtained at the end of the execution of $\AlgProcessLeftBranch(i)$ or $\AlgProcessRightBranch(i)$, respectively, consists of the vertices of $C_m$ and the vertices of the corresponding branch.

First we introduce the concept of moving the borders and we state two preliminary lemmas.
We say that $C_j$ \emph{moves} the right border of $C_{j-1}$ if $l(\rborder(C_j))>l(\rborder(C_{j-1}))$. Similarly, $C_j$ \emph{moves} the left border of $C_{j-1}$ if $r(\lborder(C_j))<r(\lborder(C_{j-1}))$.
\begin{lemma} \label{lem:proper_moves}
If $\lbranch(C_{j'},t)$, $t\leq r(\lborder(C_{j'}))$, (respectively $\rbranch(C_{j'},t)$, $t\geq l(\rborder(C_{j'}))$) is proper, then each expansion $C_j$ computed by $\AlgProcessLeftBranch(t)$ ($\AlgProcessRightBranch(t)$) moves the left (right) border of $C_{j-1}$, where $C_{j'}$ is the expansion from the beginning of the execution of the corresponding procedure $\AlgProcessLeftBranch$ or $\AlgProcessRightBranch$.
\end{lemma}
\begin{proof}
Let $C_j$ be an expansion constructed in an iteration of $\AlgProcessRightBranch$ and the other case is symmetric. In each iteration of $\AlgProcessRightBranch$ the input integer $i$ passed to $\ER$ satisfies $i=l(\rborder(C_{j-1}))$. By the formulation of $\AlgProcessRightBranch$, $i<t$. Thus, since $\rbranch(C_{j'},t)$ is proper, the vertices in $V_{i}\cap\rborder(C_{j-1})$ have no neighbors in $V_{i-1}\setminus C_{j-1}$. Hence, the instructions in Step~$\ER$ imply that $V_{i}\cap\rborder(C_j)=\emptyset$. Therefore, $l(\rborder(C_j))>i$, i.e. $C_j$ moves the right border of $C_{j-1}$.
\end{proof}
The above, and the fact that the branches $\branch_i$, $i=1,2,3$, computed by $\AlgConnectedPathwidth$ are proper, give the following.
\begin{lemma} \label{lem:moving_borders}
If $C_j$, $j\in\{2,\ldots,m\}$, is an expansion calculated in Step~$\EL$ (Step~$\ER$) invoked by $\AlgProcessLeftBranch$ ($\AlgProcessRightBranch$, respectively), then $C_j$ moves the left (right, resp.) border of $C_{j-1}$.
\qed
\end{lemma}

\begin{lemma} \label{lem:one_branch}
Let $C_{j'+1}$ and $C_{j}$ be, respectively, the first and the last expansions computed by the procedure $\AlgProcessLeftBranch(t)$ or $\AlgProcessRightBranch(t)$. If $\lbranch(C_{j'},t)$ is proper and $t\leq r(\lborder(C_{j'}))$, then the execution of $\AlgProcessLeftBranch(t)$ results in $C_{j}=C_{j'}\cup V(\lbranch(C_{j'},t))$. If $\rbranch(C_{j'},t)$ is proper and $t\geq l(\rborder(C_{j'}))$, then the execution of $\AlgProcessRightBranch(t)$ results in $C_{j}=C_{j'}\cup V(\rbranch(C_{j'},t))$.
\end{lemma}
\begin{proof}
We consider the call to $\AlgProcessLeftBranch(t)$, the proof being analogous in the other case.
Suppose that $u\in V_i$ and $v\in V_{i'}\cap\lborder(C_{j'})$, where $t\leq i\leq i'$, are connected by a progressive path $P$ in $\cG$. By the formulation of $\AlgProcessLeftBranch$ and by Lemma~\ref{lem:proper_moves}, $V(P)\subseteq C_j$. Thus, by the definition of a branch, $C_{j'}\cup V(\lbranch(C_{j'},t))\subseteq C_j$. It follows directly from the formulation of $\AlgProcessLeftBranch$ that $C_j\subseteq C_{j'}\cup V(\lbranch(C_{j'},t))$.
\end{proof}

\begin{lemma} \label{lem:branches}
Let $C_{s_0}$ be an expansion from the beginning of an iteration of $\AlgConnectedPathwidth$, and let $C_{s_i}$, $i=1,2,3$, be the expansions obtained at the end of Steps L.1, L.2 and RL.3 or R.1, R.2 and RL.3 in this iteration, respectively. Then, $C_{s_1}=C_{s_0}\cup V(\branch_1)$, $C_{s_2}=C_{s_1}\cup A_{s_2}$ and $C_{s_3}=C_{s_2}\cup V(\branch_2')\cup V(\branch_3')$. Moreover, $C_{s_3}\neq C_{s_0}$.
\end{lemma}
\begin{proof}
The equalities follow directly from Lemma~\ref{lem:one_branch} and from the formulation of the algorithm $\AlgConnectedPathwidth$.

Since the analysis in the cases $\omega(\lborder(C_{s_0}))>\omega(\rborder(C_{s_0}))$ and $\omega(\lborder(C_{s_0}))\leq\omega(\rborder(C_{s_0}))$ is similar, we assume that the former occurs. Thus, the Steps~L.1, L.2, RL.3 of $\AlgConnectedPathwidth$ are executed. If $t_1<r(\lborder(C_{s_0}))$ in Step~L.1, then by construction $V(\branch_1)\setminus C_{s_0}\neq\emptyset$ and Lemma~\ref{lem:one_branch} implies that $C_{s_3}\neq C_{s_0}$. Otherwise, that is, if $t_1=r(\lborder(C_{s_0}))$ in Step~L.1, then there exists an external vertex $v\in V_{t_1}\cap C_{s_0}$ adjacent to $u\in V_{t_1+1}\setminus C_{s_0}$. By the formulation of Step~L.2 of $\AlgConnectedPathwidth$, $u\in A_{s_2}$, and consequently $u\in C_{s_2}$, which gives $C_{s_3}\neq C_{s_0}$.
\end{proof}

Now observe that the procedure $\AlgConnectedPathwidth$ is a special case of $\AlgSimpleConnectedPathwidth$ in the sense that if we execute $\AlgConnectedPathwidth$ for the given $G$ and $\cP$ and we take the history of the executions of Steps~$\EL$ and~$\ER$, then, due to the fact that the choices between the steps S.1-S.4 in $\AlgSimpleConnectedPathwidth$ are arbitrary, it is possible to obtain exactly the same chain of executions of Steps~$\EL$ and~$\ER$ in $\AlgSimpleConnectedPathwidth$. Moreover, Lemma~\ref{lem:branches} implies that $\AlgConnectedPathwidth$ stops and returns $\cC$. Therefore, by~Lemma~\ref{lem:scp_correctness}, we obtain the following.
\begin{lemma} \label{lem:correctness}
The execution of the procedure $\AlgConnectedPathwidth$ stops, and, for the given input $G$ and $\cP$, $\AlgConnectedPathwidth$ returns a connected path decomposition of $G$.
\qed
\end{lemma}

Figure~\ref{fig:execution} gives an example of the execution of $\AlgConnectedPathwidth$. In particular, Figure~\ref{fig:execution}(a) presents a graph $\cG$ and $C_3$ (this is the expansion obtained at the end of initialization of $\AlgConnectedPathwidth$, where the vertex $v$ from Step~I.1 is the one with the weight $2$ in $V_1$). Figures~\ref{fig:execution}(b)-(d) depict the state of the algorithm at the end of the first three iterations. (The fourth iteration executes Steps L.1, L.2 and~RL.3, which ends the computation.)
\begin{figure}[ht]
	\begin{center}
	\includegraphics[scale=1.3]{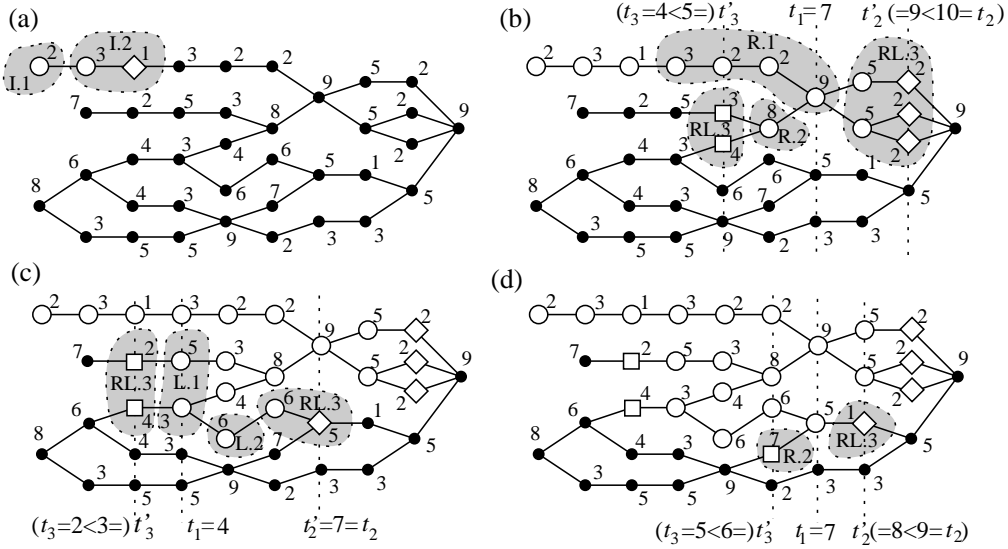}
	\caption{a graph $\cG$ (the integers are vertex weights) with white vertices in $C_m$ representing the state of $\AlgConnectedPathwidth$ after:
		(a) the initialization;
                (b) first iteration (Steps~R.1, R.2, RL.3);
                (c) second iteration (Steps L.1, L.2, RL.3);
                (c) third iteration (Steps R.1, R.2, RL.3);
                (in all cases the shaded area covers the vertices added to the current expansion during the particular step)}
	\label{fig:execution}
	\end{center}
\end{figure}

\section{The approximation guarantee of the algorithm}
\label{sec:bound}

In this section we analyze the width of the path decomposition $\cC$ calculated by $\AlgConnectedPathwidth$ for the given $G$ and $\cP$. First we introduce the concept of nested expansion, which, informally speaking, is as follows. The first condition for $C$ to be nested states that the weight of $V_i\cap C$ for any $i$ `between' the right extremity of the left border and the left extremity of the right border (by Lemma~\ref{lem:borders_separated} the former is less than the latter in each expansion computed by $\AlgConnectedPathwidth$) is greater than or equal to the weight of the left or the right border of $C$. The remaining conditions refer to the situation `inside' the borders and are symmetric for the left and right borders. Condition (ii) for the left border requires that the weight of $V_i\cap C$, where $i\leq r(\lborder(C))$, is not less than the weight of the left border restricted to the vertices in $V_1\cup\cdots\cup V_i$. Finally, condition (iii) gives a `local' minimality. Suppose that we take any left branch $\lbranch(C,i)$ (where $i$ by the definition is $\leq r(\lborder(C))$) and we add the vertices of this branch to $C$ in the way it is done in procedure $\AlgProcessLeftBranch$, then we `arrive' at some cut of this branch. Then, (iii) for $C$ guarantees that the weight of the left border of the new expansion, i.e. $\omega(\lborder(C\cup V(\lbranch(C,i))))$, is greater than or equal to the weight of the left border of $C$.

Formally, we say that an expansion $C$ is \emph{nested} if it satisfies the following conditions:
\begin{itemize}
\item[(i)]  $\omega(V_i\cap C)\geq\min\{\omega(\lborder(C)),\omega(\rborder(C))\}$ for each $i=r(\lborder(C)),\ldots,l(\rborder(C))$,
\item[(ii)] $\omega(V_i\cap C)\geq\sum_{p\leq i}\omega(V_p\cap\lborder(C))$ for each $i\leq r(\lborder(C))$, and $\omega(V_i\cap C)\geq\sum_{p\geq i}\omega(V_p\cap\rborder(C))$ for each $i\geq l(\rborder(C))$,
\item[(iii)] $r(\lborder(C))$ is a bottleneck of each branch $\lbranch(C,i)$, where $i\leq r(\lborder(C))$, and $l(\rborder(C))$ is a bottleneck of each branch  $\rbranch(C,i)$, where $i\geq l(\rborder(C))$.
\end{itemize}

Figure~\ref{fig:nested} presents a subgraph of $\cG$ induced by the vertices that belong to an expansion $C$, and with distinguished left and right borders, $\lborder(C)=\{x_1,\ldots,x_4\}$, $\rborder(C)=\{y_1,\ldots,y_4\}$. In this example $r(\lborder(C))=i+5$ and $l(\rborder(C))=i+8$. If this expansion is nested, then it holds in particular: condition (ii) implies $\omega(V_{i'}\cap C)\geq\omega(\{x_1,x_2,x_3\})$ for each $i'=i+1,\ldots,i+4$, $\omega(V_{i+5}\cap C)\geq\omega(\{x_1,x_2,x_3,x_4\})$; condition (i) implies $\omega(V_{i+6}\cap C)\geq\min\{\omega(\lborder(C)),\omega(\rborder(C))\}=\min\{\omega(\{x_1,\ldots,x_4\}),\omega(\{y_1,\ldots,y_4\})\}$.
\begin{figure}[ht]
	\begin{center}
	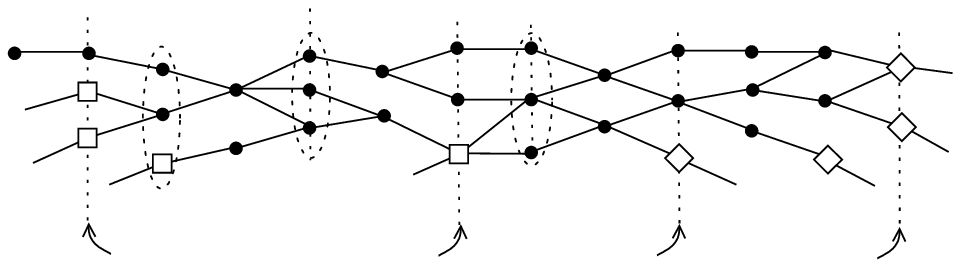
	\caption{An example of a nested expansion}
	\label{fig:nested}
	\end{center}
\end{figure}

Not all expansions computed by $\AlgConnectedPathwidth$ are nested, but we prove that all of them satisfy (ii) (see Lemmas~\ref{lem:moving_inside_borders}-\ref{lem:nested}). In particular we argue that if an expansion from the beginning of an iteration of $\AlgConnectedPathwidth$ is nested, then each expansion computed in this iteration satisfies~(ii) (Lemma~\ref{lem:condition(ii)}). Moreover, an expansion obtained at the end of this iteration is nested (Lemma~\ref{lem:nested}), which allows us to apply an induction on the number of iterations to prove the claim. We also prove that each expansion obtained in Step~L.2 and in Step~R.2 of $\AlgConnectedPathwidth$ satisfies~(i). Those facts are used in Lemma~\ref{lem:pw_bounds_borders} to prove that $\omega(B_j)\leq 2\cdot\width(\cG)$ for each $j=1,\ldots,m$. Note that $\omega(B_j)=|Z_j|$ for each $j=1,\ldots,m$. Finally, we give the main results in Theorems~\ref{thm:algorithm} and~\ref{thm:bound}. We start with two preliminary lemmas that analyze how the borders change while $\AlgProcessLeftBranch$ and $\AlgProcessRightBranch$ execute.

\begin{lemma} \label{lem:lr_borders_correct}
$\border(C_j)=\lborder(C_j)\cup\rborder(C_j)$ for each $j=1,\ldots,m$.
\end{lemma}
\begin{proof}
$\lborder(C_j)\cup\rborder(C_j)\subseteq\border(C_j)$ follows directly from Steps~$\EL$ and~$\ER$. To prove that $\border(C_j)\subseteq\lborder(C_j)\cup\rborder(C_j)$ we use induction on $j$. For $j=1$ the lemma follows directly from Step~I.1 of $\AlgConnectedPathwidth$.

Suppose that $C_j$, $1\leq j<m$, satisfies the hypothesis. Expansion $C_{j+1}$ is constructed in Step~$\EL$ or~$\ER$. Both cases are analogous so assume that the computation occurs in Step~$\EL$. Let $x\in\border(C_{j+1})$. By construction and by the induction hypothesis, $x\in\border(C_j)\cup A_{j+1}=\lborder(C_j)\cup\rborder(C_j)\cup A_{j+1}$. If $x\in\lborder(C_j)\cup A_{j+1}$, then $x\in\lborder(C_{j+1})$, because (by the definition) $\lborder(C_{j+1})=(\lborder(C_j)\cup A_{j+1})\cap\border(C_{j+1})$ and, by assumption, $x\in\border(C_{j+1})$. If $x\in\rborder(C_j)$, then $x\in\rborder(C_{j+1})$, because $\rborder(C_{j+1})=\rborder(C_j)\cap\border(C_{j+1})$.
\end{proof}

\begin{lemma} \label{lem:borders_separated}
$r(\lborder(C_j))<l(\rborder(C_j))$ for each $j=1,\ldots,m$.
\end{lemma}
\begin{proof}
To prove the lemma we use induction on $j$. The claim clearly holds for $j=1$, and consequently, by Lemma~\ref{lem:moving_borders}, it holds for all expansions obtained in the initialization stage of $\AlgConnectedPathwidth$.

Suppose that $C_j$, $1\leq j<m$, satisfies the hypothesis. The first case to consider is when the expansion $C_{j+1}$ is constructed in Step~$\EL$ or~$\ER$ called by $\AlgProcessLeftBranch$ or $\AlgProcessRightBranch$, respectively. Due to symmetry assume that the former occurs. By Lemma~\ref{lem:moving_borders}, $C_{j+1}$ moves the left border of $C_j$. Thus, by the induction hypothesis, $r(\lborder(C_{j+1}))<r(\lborder(C_j))<l(\rborder(C_j))\leq l(\rborder(C_{j+1}))$. The last inequality is due to $\rborder(C_{j+1})\subseteq\rborder(C_j)$, which follows from the definition, $\rborder(C_{j+1})=\rborder(C_j)\cap\border(C_{j+1})$.

The second case occurs when $C_{j+1}$ is computed in Step~L.2 or~R.2 of $\AlgConnectedPathwidth$.
(We consider Step~R.2, as the other case is similar.) By construction, $\lborder(C_{j+1})\subseteq\lborder(C_j)\cup V_{t_1-1}$. It holds $l(\rborder(C_j))\geq t_1$. By the induction hypothesis, $r(\lborder(C_j))<l(\rborder(C_j))$. Thus, $r(\lborder(C_{j+1}))\leq\max\{r(\lborder(C_j)),t_1-1\}<l(\rborder(C_j))\leq l(\rborder(C_{j+1}))$.
\end{proof}

In order to simplify the statements denote $\wC=C\setminus\border(C)$ for an expansion $C$.  Note that (ii) is equivalent to
\begin{itemize}
\item[(ii')] $\omega(V_i\cap\wC)\geq\sum_{p<i}\omega(V_p\cap\lborder(C))$ for each $i\leq r(\lborder(C))$, and $\omega(V_i\cap\wC)\geq\sum_{p>i}\omega(V_p\cap\rborder(C))$ for each $i\geq l(\rborder(C))$,
\end{itemize}
because, by Lemmas~\ref{lem:lr_borders_correct} and~\ref{lem:borders_separated}, $V_i\cap C$ is the sum of disjoint sets $V_i\cap\wC$ and $V_i\cap\lborder(C)$ when $i\leq r(\lborder(C))$ and $V_i\cap\wC$ and $V_i\cap\rborder(C)$ when $i\geq l(\rborder(C))$ for each expansion $C$.

\begin{lemma} \label{lem:moving_inside_borders}
Let $j\in\{2,\ldots,m\}$. If $C_{j-1}$ satisfies \textup{(ii)} and $C_j$ has been computed by the procedure $\AlgProcessLeftBranch$ or $\AlgProcessRightBranch$, then $C_{j}$ satisfies \textup{(ii)}.
\end{lemma}
\begin{proof}
All expansions $C_j$ obtained in the initialization stage of $\AlgConnectedPathwidth$ satisfy (ii), because $\lborder(C_j)=\emptyset$ and $\rborder(C_j)$ is contained in a single set $V_i$ for some $i\in\{1,\ldots,d\}$.
Assume in the following without loss of generality that $C_j$ has been computed by $\AlgProcessLeftBranch$, that is, in Step~$\EL$ and the proof of the other case is analogous. By construction, $\rborder(C_j)=\rborder(C_{j-1})\cap\border(C_j)\subseteq\rborder(C_{j-1})$. Thus, $l(\rborder(C_j))\geq l(\rborder(C_{j-1}))$, and therefore by (ii) for $C_{j-1}$ we obtain that
\begin{equation} \label{eq:r_border_correct}
\omega(V_i\cap C_j)=\omega(V_i\cap C_{j-1})\geq\sum_{p\geq i}\omega(V_p\cap\rborder(C_{j-1}))\geq\sum_{p\geq i}\omega(V_p\cap\rborder(C_{j})),
\end{equation}
for each $i\geq l(\rborder(C_j))$.

Now we prove the condition (ii) for the left border of $C_j$. By Lemma~\ref{lem:moving_borders}, $C_j$ moves the left border of $C_{j-1}$. Thus,
\begin{equation} \label{eq:moving_border1}
r(\lborder(C_j))\leq r(\lborder(C_{j-1}))-1.
\end{equation}
It follows from $\AlgProcessLeftBranch$ and from the definition of $\cG$ that for each $p<r(\lborder(C_{j}))-1$ it holds
\begin{equation} \label{eq:moving_border2}
V_p\cap C_j=V_p\cap C_{j-1}\textup{ and }V_p\cap\lborder(C_j)=V_p\cap\lborder(C_{j-1}),
\end{equation}
while for $p=r(\lborder(C_j))-1$ we have
\begin{equation} \label{eq:moving_border3}
V_p\cap C_j=V_p\cap C_{j-1}\textup{ and }V_p\cap\lborder(C_j)\subseteq V_p\cap\lborder(C_{j-1})
\end{equation}
(see Figure~\ref{fig:moving_border} for a case when $V_p\cap\lborder(C_j)\neq V_p\cap\lborder(C_{j-1})$ for $p=r(\lborder(C_j))-1$).
Thus, by (ii) for $C_{j-1}$, and by~(\ref{eq:moving_border1}),~(\ref{eq:moving_border2}),~(\ref{eq:moving_border3}), we obtain that for each $i<r(\lborder(C_j))$,
\begin{equation} \label{eq:(ii)C_j-1}
\omega(V_{i}\cap C_j)=\omega(V_{i}\cap C_{j-1})\geq\sum_{p\leq i}\omega(V_p\cap\lborder(C_{j-1}))\geq\sum_{p\leq i}\omega(V_p\cap\lborder(C_j)).
\end{equation}
\begin{figure}[ht]
	\begin{center}
	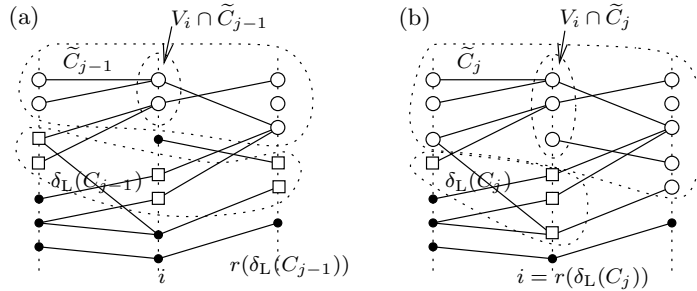
	\caption{Proof of Lemma~\ref{lem:moving_inside_borders}: the subgraph $\cG[V_{i-1}\cup V_i\cup V_{i+1}]$, where $i=r(\lborder(C_j))$ with white vertices in: (a) $C_{j-1}$, and (b) $C_j$}
	\label{fig:moving_border}
	\end{center}
\end{figure}
Let $i=r(\lborder(C_j))$. It holds $V_i\cap\wC_{j-1}\subseteq V_i\cap\wC_j$ (see also Figure~\ref{fig:moving_border}), which gives $\omega(V_i\cap\wC_{j})\geq\omega(V_i\cap\wC_{j-1})$. Thus, by~(ii'),(\ref{eq:moving_border2}),(\ref{eq:moving_border3}), by (ii) for $C_{j-1}$, and by Lemma~\ref{lem:borders_separated},
\begin{multline}
\omega(V_i\cap C_j)=\omega(V_i\cap\wC_j)+\omega(V_i\cap\lborder(C_j))\geq\omega(V_i\cap\wC_{j-1})+\omega(V_i\cap\lborder(C_j)) \\
 \geq\sum_{p<i}\omega(V_p\cap\lborder(C_{j-1}))+\omega(V_i\cap\lborder(C_j))\geq\sum_{p\leq i}\omega(V_p\cap\lborder(C_j)). \label{eq:(ii)_C_j-12}
\end{multline}
Equations~(\ref{eq:r_border_correct}),~(\ref{eq:(ii)C_j-1}) and~(\ref{eq:(ii)_C_j-12}) prove (ii) for $C_j$.
\end{proof}

\begin{lemma} \label{lem:condition(i)}
Let $C_{s_0}$ be the expansion from the beginning of an iteration of $\AlgConnectedPathwidth$ and let $C_{s_1}$ be the expansion obtained at the end of Step~L.1 or Step~R.1 of this iteration. If $C_{s_0}$ is nested, then $C_{s_1}$ satisfies \textup{(i)}.
\end{lemma}
\begin{proof}
Since $\AlgConnectedPathwidth$ executes Step~L.1 or Step~R.1 depending on the condition checked at the beginning of an iteration, and the analysis in both cases is similar, we assume without loss of generality that Step~L.1 is executed in this particular iteration of $\AlgConnectedPathwidth$. Moreover, if $t_1=r(\lborder(C_{s_0}))$ in Step~L.1, then we are done, because $s_0=s_1$ in such case.

By Lemma~\ref{lem:moving_borders}, each iteration of $\AlgProcessLeftBranch$ moves the left border of the corresponding expansion, which implies that $r(\lborder(C_{s_1}))<r(\lborder(C_{s_1-1}))<\cdots<r(\lborder(C_{s_0}))$. Moreover, by Lemma~\ref{lem:borders_separated}, $r(\lborder(C_{s_0}))<l(\rborder(C_{s_0}))$ and Lemma~\ref{lem:branches} implies that $V_i\cap\border(C_{s_1})=V_i\cap\border(C_{s_0})$ for each $i\geq r(\lborder(C_{s_0}))$, which gives
\begin{equation} \label{eq:right_unchanged}
\rborder(C_{s_0})=\rborder(C_{s_0+1})=\cdots=\rborder(C_{s_1}).
\end{equation}

Let first $i\in\{r(\lborder(C_{s_1})),\ldots,r(\lborder(C_{s_0}))\}$. By (ii') for $C_{s_0}$, and by Lemma~\ref{lem:branches}
\begin{equation} \label{eq:C_s_1_cut}
\omega(V_i\cap C_{s_1})=\omega(V_i\cap\wC_{s_0})+\omega(V_i\cap V(\branch_1))\geq\sum_{p<i}\omega(V_p\cap\lborder(C_{s_0}))+\omega(V_i\cap V(\branch_1)).
\end{equation}
By (iii) for $C_{s_0}$, the cut $r(\lborder(C_{s_0}))$ is a bottleneck of the branch $\lbranch(C_{s_0},i)$. Thus, the weight of the cut $r(\lborder(C_{s_0}))$, that is $\omega(\lborder(C_{s_0}))$, is not greater than the weight of the cut $i$ of the branch $\lbranch(C_{s_0},i)$, which equals $\omega(\Ext(\lbranch(C_{s_0},i)))$. Since $V_i\cap V(\branch_1)=V_i\cap V(\lbranch(C_{s_0},i))$, by Observation~\ref{obs:o2} and by~(\ref{eq:C_s_1_cut}) we obtain $\omega(V_i\cap C_{s_1})\geq\omega(\lborder(C_{s_0}))$.
By the fact that Step~L.1 of $\AlgConnectedPathwidth$ executes, $\omega(\lborder(C_{s_0}))\geq\omega(\rborder(C_{s_0}))$. Since, by (\ref{eq:right_unchanged}), $\omega(\rborder(C_{s_0}))=\omega(\rborder(C_{s_1}))$, we obtain that $\omega(V_i\cap C_{s_1})\geq\omega(\rborder(C_{s_1}))$.

Let now $i\in\{r(\lborder(C_{s_0})),\ldots,l(\rborder(C_{s_1}))\}$. As argued above, $\omega(\rborder(C_{s_0}))\leq\omega(\lborder(C_{s_0}))$. Thus, by~(\ref{eq:right_unchanged}), by Lemma~\ref{lem:branches} and by~(i) for $C_{s_0}$, $\omega(V_i\cap C_{s_1})=\omega(V_i\cap C_{s_0})\geq\omega(\rborder(C_{s_0}))=\omega(\rborder(C_{s_1}))$.
\end{proof}

\begin{lemma} \label{lem:condition(ii)}
If an expansion $C_{s_0}$ from the beginning of an iteration of $\AlgConnectedPathwidth$ is nested, then each expansion computed by $\AlgConnectedPathwidth$ in this iteration satisfies \textup{(ii)}.
\end{lemma}
\begin{proof}
As before, assume without loss of generality that $\AlgConnectedPathwidth$ executes Steps~L.1-L.2 in the iteration we consider, and the proof for the other case is similar.

Due to Lemma~\ref{lem:moving_inside_borders}, it is enough to prove that an expansion $C_{s_2}$ obtained in Step~L.2 of $\AlgConnectedPathwidth$ satisfies (ii). Note that $C_{s_1}=C_{s_2-1}$ is the expansion obtained at the end of Step~L.1 of $\AlgConnectedPathwidth$. If Step~$\EL$ is executed at least once by $\AlgProcessLeftBranch$ invoked in Step~L.1 of $\AlgConnectedPathwidth$, then, by Lemmas~\ref{lem:moving_inside_borders} and~\ref{lem:condition(i)}, $C_{s_1}$ satisfies~(i) and~(ii), otherwise $C_{s_1}$ is nested by assumption.

By the fact that $r(\lborder(C_{s_0}))$ is a bottleneck of $\branch_1$ from Step~L.1, we obtain that $\omega(\lborder(C_{s_1}))\geq\omega(\lborder(C_{s_0}))$. By assumption, $\omega(\lborder(C_{s_0}))\geq\omega(\rborder(C_{s_0}))$. By Lemma~\ref{lem:branches}, $\rborder(C_{s_0})=\rborder(C_{s_0+1})=\cdots=\rborder(C_{s_1})$. Thus,
\begin{equation} \label{eq:borders_of_C_s_1}
\omega(\lborder(C_{s_1}))\geq\omega(\rborder(C_{s_1})).
\end{equation}

First we prove (ii) for the left border of $C_{s_2}$. (See Figure~\ref{fig:conditionii} for an example of $C_{s_1}$ and $C_{s_2}$.) By construction, $r(\lborder(C_{s_2}))\leq t_1$. Since $C_{s_2}\subseteq C_{s_1}\cup V_{t_1+1}$, it holds $V_{t_1}\cap\lborder(C_{s_2})\subseteq V_{t_1}\cap\lborder(C_{s_1})$ and $V_i\cap\lborder(C_{s_2})=V_i\cap\lborder(C_{s_1})$ for $i<t_1$. Moreover, $V_i\cap C_{s_2}=V_i\cap C_{s_1}$ for each $i\leq t_1$. Thus, for $i\leq t_1$, by (ii) for $C_{s_1}$,
\begin{equation} \label{eq:left_ok}
\omega(V_i\cap C_{s_2})=\omega(V_i\cap C_{s_1})\geq\sum_{p\leq i}\omega(V_p\cap\lborder(C_{s_1}))\geq\sum_{p\leq i}\omega(V_p\cap\lborder(C_{s_2})).
\end{equation}

Now we analyze the right border of $C_{s_2}$. It holds $l(\rborder(C_{s_2}))\geq t_1+1$, because $\rborder(C_{s_2})\subseteq \rborder(C_{s_1})\cup V_{t_1+1}$ and $l(\rborder(C_{s_1}))\geq t_1+1$. Also,
\begin{equation} \label{eq:unchanged}
V_i\cap C_{s_2}=V_i\cap C_{s_1}\textup{ for each }i>t_1+1.
\end{equation}

\begin{figure}[ht]
	\begin{center}
	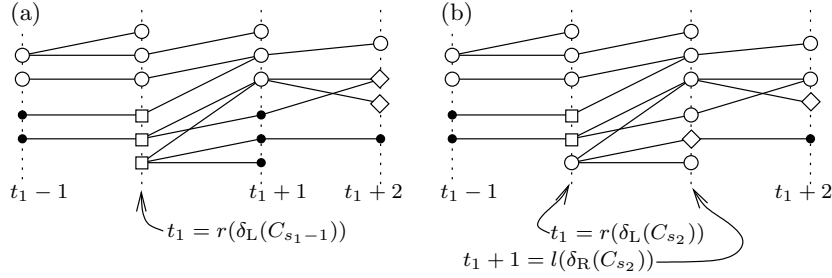
	\caption{Proof of Lemma~\ref{lem:condition(ii)}: (a) $C_{s_1}$; (b) $C_{s_2}$}
	\label{fig:conditionii}
	\end{center}
\end{figure}

Note that $\rborder(C_{s_2})\subseteq \rborder(C_{s_1})\cup V_{t_1+1}$ implies that if $l(\rborder(C_{s_2}))>t_1+1$, then $\rborder(C_{s_2})\subseteq\rborder(C_{s_1})$ and consequently, by~(\ref{eq:unchanged}), $C_{s_2}$ satisfies (ii). Thus, we continue with the assumption that $l(\rborder(C_{s_2}))=t_1+1$. Let $i\geq l(\rborder(C_{s_2}))$ be selected arbitrarily. We consider the following cases.
\begin{enumerate}
\item $i=l(\rborder(C_{s_2}))=t_1+1$. By the definition and by Lemmas~\ref{lem:lr_borders_correct} and~\ref{lem:borders_separated}, $V_{i}\cap\rborder(C_{s_2})$ and $V_{i}\cap\wC_{s_2}$ are disjoint and their union is $V_i\cap C_{s_2}$. By construction, $V_{i}\cap\wC_{s_2}\supseteq V_{i}\cap\wC_{s_1}$ (see Figure~\ref{fig:conditionii} for an example where the equality between the sets does not hold), which implies
\begin{equation} \label{eq:C_s_2_subseteqC_s_2_1}
\omega(V_i\cap\wC_{s_2})\geq\omega(V_i\cap\wC_{s_1}).
\end{equation}
By Lemma~\ref{lem:borders_separated}, $l(\rborder(C_{s_1}))>t_1$, and $r(\lborder(C_{s_1}))\leq t_1$. Thus, we obtain by (i) for $C_{s_1}$ that $\omega(V_i\cap C_{s_1})\geq\min\{\omega(\lborder(C_{s_1})),\omega(\rborder(C_{s_1}))\}$. By~(\ref{eq:borders_of_C_s_1}), $\omega(V_i\cap C_{s_1})\geq\omega(\rborder(C_{s_1}))$. By (ii') for $C_{s_1}$ and by (\ref{eq:C_s_2_subseteqC_s_2_1})
\begin{eqnarray}
\omega(V_i\cap C_{s_2}) & = & \omega(V_{i}\cap\rborder(C_{s_2}))+\omega(V_{i}\cap\wC_{s_2}) \nonumber \\
   & \geq & \omega(V_i\cap\rborder(C_{s_2}))+\omega(V_i\cap\wC_{s_1}) \nonumber\\
   & =    & \omega(V_i\cap(\rborder(C_{s_2})\setminus\rborder(C_{s_1})))+\omega(V_i\cap C_{s_1}) \nonumber\\
   & \geq & \omega(V_{i}\cap(\rborder(C_{s_2})\setminus\rborder(C_{s_1})))+\omega(\rborder(C_{s_1})) \nonumber \\
   & =    & \omega(\rborder(C_{s_2})), \nonumber
\end{eqnarray}
where the last equation follows from~(\ref{eq:unchanged}).
\item $l(\rborder(C_{s_2}))<i\leq l(\rborder(C_{s_1}))$. Since $i>t_1+1>r(\lborder(C_{s_1}))$, again by~(\ref{eq:borders_of_C_s_1}),(\ref{eq:unchanged}) and by (i),
\[\omega(V_i\cap C_{s_2})=\omega(V_i\cap C_{s_1})\geq\omega(\rborder(C_{s_1}))\geq\sum_{p\geq i}V_p\cap\rborder(C_{s_2}).\]
\item $i>l(\rborder(C_{s_1}))$. Since $C_{s_1}$ satisfies (ii), Equation~(\ref{eq:unchanged}) implies (ii) for $C_{s_2}$.
\end{enumerate}
Cases~1-3 and~(\ref{eq:left_ok}) imply that (ii) holds for $C_{s_2}$.
\end{proof}

\begin{lemma} \label{lem:nested}
Let $C_{s_0}$ and $C_{s_3}$ be the expansions from the beginning of two consecutive iterations of $\AlgConnectedPathwidth$. If $C_{s_0}$ is nested, then $C_{s_3}$ is nested.
\end{lemma}
\begin{proof}
As in the previous proofs, we continue without loss of generality with the assumption that $\omega(\lborder(C_{s_0}))>\omega(\rborder(C_{s_0}))$ at the beginning of the iteration of $\AlgConnectedPathwidth$ we consider. Let $C_{s_1}$ and $C_{s_2}$ be the expansions obtained at the end of Steps~L.1 and~L.2 of $\AlgConnectedPathwidth$, respectively.

First we prove that the expansion $C_{s_3}$ satisfies (i). It holds $r(\lborder(C_{s_3}))\leq t_3'\leq t_1<t_2'\leq l(\rborder(C_{s_3}))$ and $t_1\leq r(\lborder(C_{s_0}))$ (see also Figure~\ref{fig:cases}). This follows from the formulation of $\AlgProcessLeftBranch$ and from the definition of a branch, because for each $i$, $t_3'<i<t_2'$, $V_i\cap\border(C_{s_3})=\emptyset$, which is a consequence of the fact that $\branch_2'$ and $\branch_3'$ are proper.

Let $i\in\{r(\lborder(C_{s_3})),\ldots,l(\rborder(C_{s_3}))\}$. We consider several cases shown in Figure~\ref{fig:lemma11}, where $C_{s_0}$ with its left and right borders is given ($C_{s_1}=C_{s_0}\cup V(\branch_1)$).
\begin{figure}[ht]
	\begin{center}
	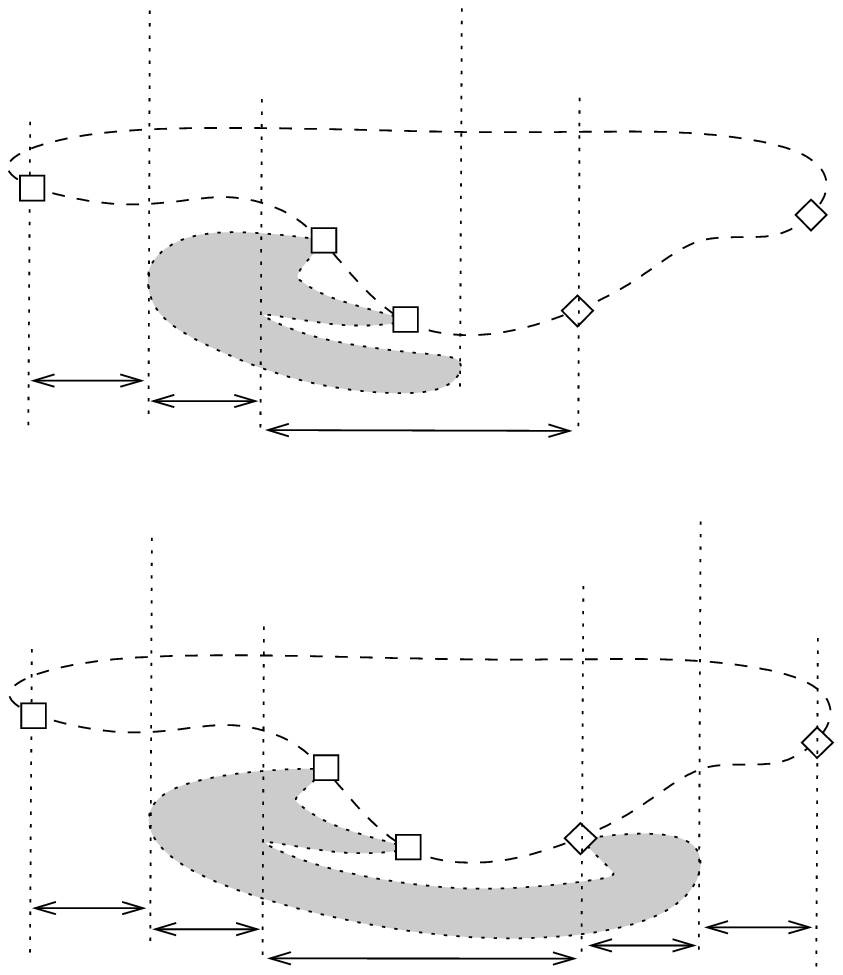
	\caption{$C_{s_0}$ together with its borders;
           (a) $t_2'\leq l(\rborder(C_{s_1}))$;
           (b) $t_2'>l(\rborder(C_{s_1}))$}
	\label{fig:lemma11}
	\end{center}
\end{figure}
Note that the example is constructed so that Case~1 and Case~5 (if applies) are `non-empty', which occurs if the branches $\branch_p'$ have no external vertices in $V_{t_p'}$, $p=2,3$. In each case we prove that $\omega(V_i\cap C_{s_3})\geq\omega(\lborder(C_{s_3}))$ or $\omega(V_i\cap C_{s_3})\geq\omega(\rborder(C_{s_3}))$.

\begin{list}{}{\leftmargin=3em}
\item[Case 1:] $r(\lborder(C_{s_3}))\leq i<t_3'$. By construction, $V_p\cap C_{s_3}=V_p\cap C_{s_0}$ for each $p<t_3'$, $V_p\cap\lborder(C_{s_3})=V_p\cap\lborder(C_{s_0})$ for each $p<t_3'-1$ and $V_{t_3'-1}\cap\lborder(C_{s_3})\subseteq V_{t_3'-1}\cap\lborder(C_{s_0})$. By (ii) for $C_{s_0}$ (note that $i<t_1\leq r(\lborder(C_{s_0}))$), $\omega(V_i\cap C_{s_0})\geq\sum_{p\leq i}\omega(V_p\cap\lborder(C_{s_0}))$. Thus,
\begin{equation} \label{eq:bound_on_C_s_0}
\omega(V_i\cap C_{s_3})=\omega(V_i\cap C_{s_0})\geq\sum_{p\leq i}\omega(V_p\cap\lborder(C_{s_3})).
\end{equation}

The fact that $r(\lborder(C_{s_3}))\leq i$ implies
\[\bigcup_{p\leq i}V_p\cap\lborder(C_{s_3})=\lborder(C_{s_3}),\]
which by~(\ref{eq:bound_on_C_s_0}) gives $\omega(V_i\cap C_{s_3})\geq\omega(\lborder(C_{s_3}))$.

\item[Case 2:] $t_3'\leq i\leq t_1$. By Lemma~\ref{lem:branches} and by the formulation of Step~L.1 of $\AlgConnectedPathwidth$, $V_i\cap C_{s_3}=(V_i\cap\wC_{s_0})\cup(V_i\cap V(\branch_3'))$. Since $\wC_{s_0}$ and $V(\branch_3')$ are (by the definition) disjoint, we obtain that $\omega(V_i\cap C_{s_3})=\omega(V_i\cap\wC_{s_0})+\omega(V_i\cap V(\branch_3'))$. By assumption, $C_{s_0}$ satisfies~(ii), and $i\leq t_1\leq r(\lborder(C_{s_0}))$. Thus, by~(ii') and by Observations~\ref{obs:o1} and~\ref{obs:o2},
\[\omega(V_i\cap C_{s_3})\geq\sum_{p<i}\omega(V_p\cap\lborder(C_{s_0}))+\omega(V_i\cap V(\branch_3'))\geq \Ext(\lbranch(C_{s_2},i))\geq\Ext(\lbranch(C_{s_2},t_3'))=\omega(\lborder(C_{s_3})),\]
because $t_3'$ is a bottleneck of $\branch_3$.

\item[Case 3:] $t_1<i\leq l(\rborder(C_{s_1}))$. By Lemma~\ref{lem:branches}, by the choice of $i$ and by the fact that $\branch_2$, and therefore $\branch_2'$, is proper, the set $V_i\cap C_{s_3}$ is an union of two disjoint sets $V_i\cap C_{s_1}$ and $V_i\cap V(\branch_2')$, which gives
\begin{equation} \label{eq:three_disjoint1}
\omega(V_i\cap C_{s_3})=\omega(V_i\cap C_{s_1})+\omega(V_i\cap V(\branch_2')).
\end{equation}
By Lemma~\ref{lem:condition(i)}, $C_{s_1}$ satisfies (i), which in particular gives $\omega(V_i\cap C_{s_1})\geq\min\{\omega(\lborder(C_{s_1})),\omega(\rborder(C_{s_1}))\}$.

If $\omega(\lborder(C_{s_1}))\geq\omega(\rborder(C_{s_1}))$, then
\[\omega(V_i\cap C_{s_1})+\omega(V_i\cap V(\branch_2'))\geq\omega(\rborder(C_{s_1}))+\omega(V_i\cap V(\branch_2'))\geq\omega(\rborder(C_{s_3})),\]
because $t_2'$ is a bottleneck of $\branch_2$.
By (\ref{eq:three_disjoint1}), $\omega(V_i\cap C_{s_3})\geq\omega(\rborder(C_{s_3}))$.

If $\omega(\lborder(C_{s_1}))<\omega(\rborder(C_{s_1}))$, then by~(\ref{eq:three_disjoint1})
\[\omega(V_i\cap C_{s_3}))\geq\omega(\lborder(C_{s_1}))+\omega(V_i\cap V(\branch_2'))\geq\omega(\lborder(C_{s_3}))+\omega(V_i\cap V(\branch_2'))\geq\omega(\lborder(C_{s_3})),\]
where $\omega(\lborder(C_{s_1}))\geq\omega(\lborder(C_{s_3}))$ follows from the fact that $t_3'$ is a bottleneck of $\branch_3$.

\item[Case 4:] $l(\rborder(C_{s_1}))<i\leq t_2'$. Note that if $l(\rborder(C_{s_1}))>t_2'$, then there is nothing to prove (see Figure~\ref{fig:lemma11}(a)). Otherwise the situation is symmetric to Case~2 and (due to $\rborder(C_{s_1})\subseteq\rborder(C_{s_0})$) leads to inequality $\omega(V_i\cap C_{s_3})\geq\omega(\rborder(C_{s_3}))$.

\item[Case 5:] $t_2' <i\leq l(\rborder(C_{s_3}))$. The proof is analogous to Case~1 and gives $\omega(V_i\cap C_{s_3})\geq\omega(\rborder(C_{s_3}))$.
\end{list}
Since in each case we obtain that $\omega(V_i\cap C_{s_3})\geq\omega(\lborder(C_{s_3}))$ or $\omega(V_i\cap C_{s_3})\geq\omega(\rborder(C_{s_3}))$, where $i\in\{r(\lborder(C_{s_3})),\ldots,l(\rborder(C_{s_3}))\}$, we have proved (i) for $C_{s_3}$.

Lemma~\ref{lem:condition(ii)} implies that $C_{s_3}$ satisfies (ii).

Note that (iii) for $C_{s_3}$ follows directly from Step RL.3 of $\AlgConnectedPathwidth$. Indeed, if $t_3'$ is not a bottleneck of any left branch $\branch_3''=\lbranch(C_{s_3},i)$, i.e. $t_3''<t_3'$ is its bottleneck, then $t_3''$ is also a bottleneck of $\branch_3$ from Step~RL.3 of $\AlgConnectedPathwidth$ (the union of $\branch_3'$ and $\branch_3''$ is a subgraph of $\branch_3$, because $\branch_3$ is maximal), which contradicts the choice of the branch $\branch_3'$. Similar argument holds for the right border of $C_{s_3}$.
\end{proof}

\begin{lemma} \label{lem:pw_bounds_borders}
It holds $\omega(B_j)\leq 2\cdot\width(\cG)$ for each $j=1,\ldots,m$.
\end{lemma}
\begin{proof}
An expansion obtained at the end of Step~I.2 of $\AlgConnectedPathwidth$ is nested. Indeed, its left border is empty which implies (i) and gives (ii) and~(iii) for the left border. Condition (iii) for the right border follows from $a_0'$ being the bottleneck of the branch used in Step~I.2, and (ii) trivially holds, because the right border is contained in a single set $V_i$. Therefore, using an induction (on the number of iterations of $\AlgConnectedPathwidth$) we obtain by Lemma~\ref{lem:nested} that any expansion from the beginning of an iteration of $\AlgConnectedPathwidth$ is nested.
Thus, Lemma~\ref{lem:condition(ii)}, implies that $C_j$ satisfies (ii) for each $j=1,\ldots,m$.

Let $j\in\{2,\ldots,m\}$. First note that $\max\{\omega(\rborder(C_j)),\omega(\lborder(C_j))\}\leq\width(\cG)$. Indeed, by (ii), where $i=l(\rborder(C))$,
\begin{equation} \label{eq:bounded_by_width}
\omega(V_i\cap C_j)\geq\sum_{p\geq i}\omega(V_p\cap\rborder(C_j))=\omega(\rborder(C_j)).
\end{equation}
Since $\omega(V_i\cap C_j)\leq\omega(V_i)\leq\width(\cG)$, we obtain $\omega(\rborder(C_j))\leq\width(\cG)$. Analogously one can prove that $\omega(\lborder(C_j))\leq\width(\cG)$.

Now we give the upper bound on $\omega(B_j)$. The claim clearly follows for $j=1$ so assume in the following that $j>1$.
By construction, $B_j=\border(C_j)\cup A_j$. Thus, due to Lemma~\ref{lem:lr_borders_correct} we obtain that if
\begin{equation} \label{eq:needed}
\omega(A_j\cup\lborder(C_{j}))\leq\width(\cG)\textup{ or }\omega(A_j\cup\rborder(C_{j}))\leq\width(\cG),
\end{equation}
then, by~(\ref{eq:bounded_by_width}), the lemma follows, so we now prove that one of those inequalities holds.

By construction, $A_j\subseteq V_i$ for some $i\in\{1,\ldots,d\}$, and
\begin{equation} \label{eq:A_j_C_j-1_disjoint}
A_j\cap(V_i\cap C_{j-1})=\emptyset,
\end{equation}
because $A_j\cap C_{j-1}=\emptyset$. We consider two cases:
\begin{list}{}{\addtolength{\leftmargin}{1em}}
\item[Case 1:] $C_j$ has been computed in Step~L.2 or in Step~R.2 of $\AlgConnectedPathwidth$. 
Assume without loss of generality that the former occurs. By construction, $i\in\{r(\lborder(C_{j-1})),\ldots,l(\rborder(C_{j-1}))\}$ and, by Lemma~\ref{lem:condition(i)}, $C_{j-1}$ satisfies~(i). Hence, by~(\ref{eq:A_j_C_j-1_disjoint}),
\begin{equation} \label{eq:C_j_by_C_j-1}
\min\{\omega(A_j\cup\lborder(C_{j-1})),\omega(A_j\cup\rborder(C_{j-1}))\}\leq\omega(A_j\cup(V_i\cap C_{j-1}))\leq\omega(V_i)\leq\width(\cG).
\end{equation}
By the formulation of Step~$\EL$, $\lborder(C_j)\subseteq A_j\cup\lborder(C_{j-1})$, which gives $A_j\cup\lborder(C_j)\subseteq A_j\cup\lborder(C_{j-1})$. Similarly, $A_j\cup\rborder(C_j)\subseteq A_j\cup\rborder(C_{j-1})$. Consequently, $\omega(A_j\cup\lborder(C_j))\leq\omega(A_j\cup\lborder(C_{j-1}))$ and $\omega(A_j\cup\rborder(C_j))\leq\omega(A_j\cup\rborder(C_{j-1}))$.
Equation~(\ref{eq:C_j_by_C_j-1}) implies (\ref{eq:needed}), which gives the desired bound on $\omega(B_j)$.

\item[Case 2:] $C_j$ has been computed by $\AlgProcessLeftBranch$ or $\AlgProcessRightBranch$.
Since both cases are analogous, assume again that the former occurs. By Lemma~\ref{lem:moving_borders}, $C_j$ moves the left border of $C_{j-1}$, which gives that $i\leq r(\lborder(C_{j-1}))$. Moreover, $\lborder(C_j)\subseteq A_j\cup\bigcup_{p\leq i}V_p\cap\lborder(C_{j-1})$. Thus, $A_j\cup\lborder(C_j)\subseteq A_j\cup\bigcup_{p\leq i}V_p\cap\lborder(C_{j-1})$, and therefore
\[\omega(A_j\cup\lborder(C_j))\leq\omega(A_j)+\sum_{p\leq i}\omega(V_p\cap\lborder(C_{j-1})).\]
Then, by~(\ref{eq:A_j_C_j-1_disjoint}) and by (ii) for $C_{j-1}$,
\[\omega(A_j)+\sum_{p\leq i}\omega(V_p\cap\lborder(C_{j-1}))\leq\omega(A_j)+\omega(V_i\cap C_{j-1})\leq\omega(V_i)\leq\width(\cG),\]
and~(\ref{eq:needed}) follows.
\end{list}
\end{proof}

\begin{lemma} \label{lem:bound}
If $\cC=(Z_1,\ldots,Z_m)$ is a path decomposition calculated by $\AlgConnectedPathwidth$ for the given $G$ and $\cP$, then $\width(\cC)\leq 2\cdot\width(\cP)+1$.
\end{lemma}
\begin{proof}
By Lemma~\ref{lem:pw_bounds_borders}, $\omega(B_j)\leq2\cdot\width(\cG)$. By the definition, $\width(\cC)=\max\{\omega(B_j)\colon j=1,\ldots,m\}-1$. Thus, by the definition, $\width(\cC)\leq2\cdot\width(\cG)-1=2\cdot\width(\cP)+1$.
\end{proof}

\begin{lemma} \label{lem:running_time}
Let $G$ be a simple connected graph and let $\cP=(X_1,\ldots,X_d)$ be its path decomposition of width $k$. The running time of $\AlgConnectedPathwidth$ executed for $G$ and $\cP$ is $O(dk^2)$.
\end{lemma}
\begin{proof}
Since each edge of $G$ is contained in one of the bags of $\cP$, $|E(G)|\leq dk^2$. The number of vertices and edges in $\cG$ is $O(kd)$ and $O(dk^2)$, respectively. Thus, the complexity of constructing $\cG$ is $O(dk^2)$.

Here we assume that each branch $\branch$ used by $\AlgConnectedPathwidth$ is encoded as a linked list such that each element of this list is a non-empty set $V_i\cap V(\branch)$, $i\in\{1,\ldots,d\}$.
If a branch is given, then the weights of all its cuts can be calculated in time linear in the number of edges and vertices of the branch. The computation of a branch is done by the execution of the procedure $\AlgProcessLeftBranch$ or $\AlgProcessRightBranch$ (due to symmetry assume that the former occurs), and if $C_j$ and $C_{j'}$ are the expansions from the beginning and from the end, respectively, of a particular execution of $\AlgProcessLeftBranch$, then the vertex set of the branch is $\lborder(C_{j})\cup(C_{j'}\setminus C_j)$ for the corresponding left branch. The time of finding any branch $\branch$ in an iteration of $\AlgConnectedPathwidth$ is therefore $O(|E(\branch)|)$. Also note that, while recording a subset $A_j\subseteq V_i$ of vertices of $\branch$ during the execution of $\AlgProcessLeftBranch$, the weight of cut $i$ of the corresponding branch can be efficiently obtained, because it is equal to $\omega(\lborder(C_j))$. Moreover, for each $j'>j$ if $C_{j'}$ has been computed in the same execution of $\AlgProcessLeftBranch$, then the weight of cut $i$ of the branch corresponding to $C_{j'}$ remains $\omega(\lborder(C_j))$. Thus, the complexity of calculating the weight of all cuts of $\branch$, and thus finding its bottleneck, is $O(|E(\branch)|)$.

Whenever two branches overlap, we do not have to repeat the computation, because due to Observation~\ref{obs:o1}, one is contained in the other, and their vertex sets and cuts are identical for common induces $i$. Therefore, the time complexity of determining all branches and their bottlenecks is $O(dk^2)$. This includes the complexity of all executions of the procedures $\AlgProcessLeftBranch$ and $\AlgProcessRightBranch$, because, by Lemma~\ref{lem:branches}, the procedure `follows' the previously calculated branches by including their vertices into the expansions $C_j$. It holds that $m\leq kd$, because  (by Lemma~\ref{lem:branches} and by the formulation of Steps~$\EL$ and~$\ER$) $C_j\subseteq C_{j+1}$ and $C_j\neq C_{j+1}$ for each $j=1,\ldots,m-1$. By Lemma~\ref{lem:bound}, $\omega(B_j)=O(k)$ for each $j=1,\ldots,m$. Thus, $\sum_{1\leq j\leq m}|Z_j|=O(dk^2)$. Therefore, the overall complexity of $\AlgConnectedPathwidth$ is $O(dk^2)$.
\end{proof}

\begin{theorem} \label{thm:algorithm}
There exists an algorithm that for the given connected graph $G$ and its path decomposition $\cP=(X_1,\ldots,X_d)$ of width $k$ returns a connected path decomposition $\cC=(Z_1,\ldots,Z_m)$ such that $\width(\cC)\leq 2k+1$ and $m\leq kd$. The running time of the algorithm is $O(dk^2)$.
\end{theorem}
\begin{proof}
The correctness of the algorithm $\AlgConnectedPathwidth$ is due to Lemma~\ref{lem:correctness}. The inequality $\width(\cC)\leq 2\cdot\width(\cP)+1$ follows from Lemma~\ref{lem:bound}, while the complexity of $\AlgConnectedPathwidth$ is due to Lemma~\ref{lem:running_time}. As argued in the proof of Lemma~\ref{lem:running_time}, $m\leq kd$.
\end{proof}

\begin{theorem} \label{thm:bound}	
For each connected graph $G$, $\cpw(G)\leq 2\cdot\pw(G)+1$.
\qed
\end{theorem}

\section{Applications in graph searching}
\label{sec:searching}

In this section we restate the main result of the previous section in terms of the graph searching numbers. In addition, we propose a small modification to the algorithm $\AlgConnectedPathwidth$ which can be used to convert a search strategy into a connected one that starts at an arbitrary homebase $h\in V(G)$. To that end it is sufficient to guarantee that $h$ belongs to the first bag of the resulting path decomposition $\cC$. The modification changes only the initialization stage of $\AlgConnectedPathwidth$.

Consider the following procedure $\AlgConnectedPathwidthWithHomebase$ (\emph{Connected Pathwidth with Homebase}) obtained by replacing the Steps~I.1-I.2 of $\AlgConnectedPathwidth$ with Steps~I.1'-I.3':

\noindent
\algrule
\f{0}{\textbf{Procedure} $\AlgConnectedPathwidthWithHomebase$ (\emph{Connected Pathwidth with Homebase})} \algrule
\f[10pt]{0}{\textbf{Input:} A simple graph $G$, a path decomposition $\cP$ of $G$, and $h\in V(G)$.}
\f[10pt]{0}{\textbf{Output:} A connected path decomposition $\cC$ of $G$ with $h$ in its first bag.}
\f{0}{\textbf{begin}}
\f[15pt]{1}{I.1': Use $G$ and $\cP$ to calculate the derived graph $\cG$. Let $v=v(H)$ be any vertex of $\cG$ such that $h\in V(H)$. Let $C_1=\{x,y\}$, where $v\in C_1$, $x,y$ are adjacent in $\cG$, and $x\in V_i$, $y\in V_{i+1}$ for some $i\in\{1,\ldots,d-1\}$. Let $\lborder(C_1)=\{x\}\cap\border(C_1)$, $\rborder(C_1)=\{y\}\cap\border(C_1)$ and let $m=1$.}
\f[15pt]{1}{I.2': If $\lborder(C_m)\neq\emptyset$, then find the maximal left branch $\lbranch(C_m,a_0)$ with a bottleneck $a_0'$ ($a_0'\geq a_0$) and call $\AlgProcessLeftBranch(a_0')$.}
\f[15pt]{1}{I.3': If $\rborder(C_m)\neq\emptyset$, then find the maximal right branch $\rbranch(C_m,b_0)$ with a bottleneck $b_0'$ ($b_0'\leq b_0$) and call $\AlgProcessRightBranch(b_0')$.}
\f{1}{\emph{Execute the main loop of $\AlgConnectedPathwidth$, compute $Z_1,\ldots,Z_m$ as in $\AlgConnectedPathwidth$ and return $\cC$.}}
\f{0}{\textbf{End procedure} $\AlgConnectedPathwidthWithHomebase$.}
\algrule

\begin{lemma} \label{lem:new_init}
Given a simple graph $G$, a path decomposition $\cP$ of $G$ and $h\in V(G)$ as an input, the algorithm $\AlgConnectedPathwidthWithHomebase$ returns a connected path decomposition $\cC=(Z_1,\ldots,Z_m)$ of $G$ such that $\width(\cC)\leq 2\cdot\width(\cP)+1$ and $h\in Z_1$.
\end{lemma}
\begin{proof}
Note that $h\in Z_1$ follows from Step~I.1' of $\AlgConnectedPathwidthWithHomebase$. If $C_j$ is an expansion obtained in one of Steps~I.1'-I.3', then $B_j\subseteq V_i\cup V_{i'}$ for some $i,i'\in\{1,\ldots,d\}$. Thus, $|Z_j|\leq 2\cdot\width(\cG)$. The expansion obtained at the end of Step~I.3' is nested, which follows from the fact that $a_0'$ and $b_0'$ are the bottlenecks of the branches computed in Steps~I.2' and~I.3'. Thus, the proof is analogous to the proof of Lemma~\ref{lem:bound}.
\end{proof}

Suppose that we are given a graph $G$, a search strategy that uses $k$ searchers, and a vertex $h\in V(G)$ that is required to be the homebase of the connected search strategy to be computed. Recall that $\nsn(G)=\pw(G)+1$, and whatsmore, a node search strategy for $G$ that uses $p$ searchers can be converted into a path decomposition of $G$ of width $p-1$ and vice versa (see \cite{Bienstock_survey91,Kinnersley92,KirousisPapadimitriou85,Mohring90}). Thus, the initial (non-connected) search strategy can be converted into a path decomposition $\cP$ of width $k$ \cite{searching_and_pebbling}.
By Lemma~\ref{lem:new_init}, the procedure $\AlgConnectedPathwidthWithHomebase$ returns a connected path decomposition $\cC$ of width $2k+1$, where $k$ is the width of the input path decomposition $\cP$. Moreover, the homebase $h$ is guaranteed to belong to the first bag $Z_1$ of $\cC$. Then, we convert $\cC$ into a monotone and connected search strategy that uses at most $2k+3$ searchers \cite{searching_and_pebbling}. The monotonicity follows directly from the definition of path decomposition. This leads to the following.

\begin{theorem} \label{thm:algorithm2}
There exists an algorithm that for a given connected graph $G$, $h\in V(G)$ and a search strategy that uses $k$ searchers returns a monotone connected search strategy with homebase $h$ that uses at most $2k+3$ searchers. The running time of the algorithm is $O(dk^2)$.
\qed
\end{theorem}
\begin{theorem} \label{thm:bound2}	
For each connected graph $G$, $\csn(G)\leq\mcsn(G)\leq 2\cdot\sn(G)+3$.
\qed
\end{theorem}

\section{Conclusions}
\label{sec:conclusions}
The advances in graph theory presented in this paper are three-fold:
\begin{itemize}
 \item[$\circ$] A bound for connected pathwidth is given, $\cpw(G)\leq 2\cdot\pw(G)+1$, where $G$ is any graph, which bounds the connected search number of a graph, $\csn(G)\leq 2\sn(G)+3$. Moreover, a vertex $v$ that belongs to the first bag in the resulting connected path decomposition can be selected arbitrarily, which implies a stronger fact, namely a connected $(2\sn(G)+3)$-search strategy can be constructed with any vertex of $G$ playing the role of the homebase. Since one can obtain a path decomposition of width $k$ for a given search strategy that uses $k$ searchers, our algorithm provides an efficient algorithm for converting a search strategy into a connected one, and in addition, the homebase in the latter one can be chosen arbitrarily.
 \item[$\circ$] An efficient method is given for calculating a connected path decomposition of width at most $2k+1$, provided that a graph $G$ and its path decomposition of width $k$ are given as an input.
 \item[$\circ$] It is a strong assumption that the algorithm requires a path decomposition to be given, because calculating $\pw(G)$ is a hard problem even for graphs $G$ that belong to some special classes of graphs. However, this algorithm can be used to approximate the connected pathwidth, because any $O(q)$-approximation algorithm for pathwidth can be used together with the algorithm from this paper to design a $O(q)$-approximation algorithm for connected pathwidth.
\end{itemize}

\noindent
{\bf Acknowledgment.} Thanks are due to the anonymous referees for constructive and insightful comments that helped to improve this work.

\bibliographystyle{plain}
\bibliography{search}

\begin{thebibliography}{10}

\bibitem{searchng_and_sweeping}
B.~Alspach.
\newblock Searching and sweeping graphs: a brief survey.
\newblock {\em Le Matematiche (Catania)}, 59:5--37, 2004.

\bibitem{Barriere_Inria2010}
L.~Barri{\`e}re, P.~Flocchini, F.V. Fomin, P.~Fraigniaud, N.~Nisse, N.~Santoro,
  and D.M. Thilikos.
\newblock Connected graph searching.
\newblock Research Report {RR}-7363, INRIA, 2010.

\bibitem{connected_weighted_trees}
L.~Barri\`{e}re, P.~Flocchini, P.~Fraigniaud, and N.~Santoro.
\newblock Capture of an intruder by mobile agents.
\newblock In {\em SPAA '02: Proc. of the fourteenth annual ACM symposium on
  parallel algorithms and architectures}, pages 200--209, New York, NY, USA,
  2002. ACM.

\bibitem{searching_not_jumping}
L.~Barri{\`e}re, P.~Fraigniaud, N.~Santoro, and D.M. Thilikos.
\newblock Searching is not jumping.
\newblock In {\em WG '03: Proc. of the 29th Inter. Workshop on Graph-Theoretic
  Concepts in Computer Science}, pages 34--45, 2003.

\bibitem{Bienstock_survey91}
D.~Bienstock.
\newblock Graph searching, path-width, tree-width and related problems (a
  survey).
\newblock {\em DIMACS Ser. Discrete Math. Theoret. Comput. Sci.}, 5:33--49,
  1991.

\bibitem{monotonicity_in_graph_searching}
D.~Bienstock and P.~Seymour.
\newblock Monotonicity in graph searching.
\newblock {\em J. Algorithms}, 12(2):239--245, 1991.

\bibitem{treewidth_guide}
H.L. Bodlaender.
\newblock A tourist guide through treewidth.
\newblock {\em Acta Cybern.}, 11(1-2):1--22, 1993.

\bibitem{deren_weighted_trees}
D.~Dereniowski.
\newblock Connected searching of weighted trees.
\newblock {\em Theor. Comp. Sci.}, 412:5700--5713, 2011.

\bibitem{price_of_connectedness}
F.V. Fomin, P.~Fraigniaud, and D.M. Thilikos.
\newblock The price of connectedness in expansions.
\newblock Technical report, Technical Report, UPC Barcelona, 2004.

\bibitem{guaranteed_graph_searching}
F.V. Fomin and D.M. Thilikos.
\newblock An annotated bibliography on guaranteed graph searching.
\newblock {\em Theor. Comput. Sci.}, 399(3):236--245, 2008.

\bibitem{connected_treewidth}
P.~Fraigniaud and N.~Nisse.
\newblock Connected treewidth and connected graph searching.
\newblock In {\em Proc.~of the 7th Latin American Symposium on Theoretical
  Informatics (LATIN'06), LNCS}, volume 3887, pages 479--490, Valdivia, Chile,
  2006.

\bibitem{monotony_properties_connected_visible}
P.~Fraigniaud and N.~Nisse.
\newblock Monotony properties of connected visible graph searching.
\newblock {\em Inf. Comput.}, 206(12):1383--1393, 2008.

\bibitem{Gustedt93}
J.~Gustedt.
\newblock On the pathwidth of chordal graphs.
\newblock {\em Discrete Appl. Math.}, 45(3):233--248, 1993.

\bibitem{pathwidth_hard}
T.~Kashiwabara and T.~Fujisawa.
\newblock Np-completeness of the problem of finding a minimum-clique-number
  interval graph containing a given graph as a subgraph.
\newblock In {\em Proc. IEEE Inter. Symp. Circuits and Systems}, pages
  657--660, 1979.

\bibitem{Kinnersley92}
N.G. Kinnersley.
\newblock The vertex separation number of a graph equals its path-width.
\newblock {\em Inf. Process. Lett.}, 42(6):345--350, 1992.

\bibitem{KirousisPapadimitriou85}
L.M. Kirousis and C.H. Papadimitriou.
\newblock Interval graphs and searching.
\newblock {\em Discrete App. Math.}, 55:181--184, 1985.

\bibitem{searching_and_pebbling}
L.M. Kirousis and C.H. Papadimitriou.
\newblock Searching and pebbling.
\newblock {\em Theor. Comput. Sci.}, 47(2):205--218, 1986.

\bibitem{LaPaugh93}
A.S. LaPaugh.
\newblock Recontamination does not help to search a graph.
\newblock {\em J. ACM}, 40(2):224--245, 1993.

\bibitem{MegiddoHGJP88}
N.~Megiddo, S.L. Hakimi, M.R. Garey, D.S. Johnson, and C.H. Papadimitriou.
\newblock The complexity of searching a graph.
\newblock {\em J. ACM}, 35(1):18--44, 1988.

\bibitem{weighted_pathwidth09}
R.~Mihai and I.~Todinca.
\newblock Pathwidth is {NP}-hard for weighted trees.
\newblock In {\em FAW '09: Proc. of the 3rd Inter. Workshop on Frontiers in
  Algorithmics}, pages 181--195, Berlin, Heidelberg, 2009. Springer-Verlag.

\bibitem{Mohring90}
R.~M\"ohring.
\newblock Graph problems related to gate matrix layout and {PLA} folding.
\newblock In {\em E. Mayr, H. Noltemeier, and M. Syslo eds, Computational Graph
  Theory, Computing Supplementum}, volume~7, pages 17--51, 1990.

\bibitem{connected_searching_chordal_graphs}
N.~Nisse.
\newblock Connected graph searching in chordal graphs.
\newblock {\em Discrete Applied Math.}, 157(12):2603--2610, 2008.

\bibitem{Peng06}
S.L. Peng, M.T. Ko, C.W. Ho, T.S. Hsu, and C.Y. Tang.
\newblock Graph searching on chordal graphs.
\newblock In {\em ISAAC '96: Proc. of the 7th Inter. Symposium on Algorithms
  and Computation}, pages 156--165, London, UK, 1996. Springer-Verlag.

\bibitem{RobertsonSeymour83}
N.~Robertson and P.D. Seymour.
\newblock Graph minors. {I}. excluding a forest.
\newblock {\em J. Comb. Theory, Ser. B}, 35(1):39--61, 1983.

\bibitem{RS_treewidth}
N.~Robertson and P.D. Seymour.
\newblock Graph minors. {II}. {A}lgorithmic aspects of tree-width.
\newblock {\em J. Algorithms}, 7(3):309--322, 1986.

\bibitem{sweeping_large_cliques}
B.~Yang, D.~Dyer, and B.~Alspach.
\newblock Sweeping graphs with large clique number.
\newblock {\em Discrete Mathematics}, 309(18):5770--5780, 2009.

\end{thebibliography}
\end{document}